\documentclass[%
reprint,
showpacs,preprintnumbers,
 amsmath,amssymb,
 aps,
pra,
]{revtex4-1}

\usepackage{graphicx}
\usepackage{dcolumn}
\usepackage{subfigure}
\usepackage{bm}

\usepackage{color,soul}
\usepackage{xcolor}	

\newtheorem{theorem}{Theorem}

\newtheorem{corollary}[theorem]{Corollary}
\newtheorem{observation}[theorem]{Observation}

\newenvironment{proof}[1][Proof]{\begin{trivlist}
\item[\hskip \labelsep {\bfseries #1}]}{\end{trivlist}}

 \newcommand{\ket}[1]{|#1\rangle}
 \newcommand{\bra}[1]{\langle #1|}
 \newcommand{\project}[1]{\ket{#1}\bra{#1}}

 \newcommand{\Id}{{\mathbb I}}
 
 \newcommand{\rank}{{\mathrm rank}}
 
 \newcommand{\T}{{\mathrm t}}

\newcommand{\Tr}{{\mathrm {Tr}}}
\newcommand{\diag}{{\mathrm {diag}}}

\newcommand{\etal}{\textit {et al.}}

\begin{document}


\title{Computable Measure of Total Quantum Correlation of Multipartite Systems}

\author{Javad Behdani}
 \email{behdani.javad@stu.um.ac.ir}
\affiliation{Department of Physics, Ferdowsi University of Mashhad, Mashhad, Iran}
\author{Seyed Javad Akhtarshenas}
 \email{akhtarshenas@um.ac.ir}
\affiliation{Department of Physics, Ferdowsi University of Mashhad, Mashhad, Iran}
\author{Mohsen Sarbishaei}
 \email{sarbishei@um.ac.ir}
\affiliation{Department of Physics, Ferdowsi University of Mashhad, Mashhad, Iran}


\begin{abstract}
Quantum discord as a measure of the quantum correlations cannot be easily computed for most of density operators.  In this paper, we present a measure of the total quantum correlations that is operationally simple and can be computed effectively for an arbitrary mixed state of a multipartite system. The measure is based on the coherence vector of the party whose quantumness is investigated  as well as  the correlation matrix of this part with the remainder of the system. Being able to detect the quantumness  of multipartite systems, such as detecting the quantum critical points in spin chains, alongside with the computability characteristic of the measure, make it a useful indicator to be exploited in the cases which are out of the scope of the other known measures.


\end{abstract}



\maketitle


\section{Introduction}

\label{s1}

The quantum nature of a state has been taken into rigorous consideration in the quantum information theory \cite{Einstein&Podolsky&Rosen:1935,Bell:1964}. For decades, entanglement was assumed to be the only source of quantum correlations, and lots of measures for this concept were introduced, such as entanglement of formation, distillable entanglement, entanglement of assistance, concurrence, and negativity \cite{Peres:1996,Bennett&etal:1996a,Bennett&etal:1996b,Bennett&etal:1996c,Vedral&etal:1997,Popescu&Rohrlich:1997,Cerf&Adami:1998,Vedral&Plenio:1998,Lewenstein&Sanpera:1998}. Entanglement is the source of many information theoretic merits in quantum systems, compared with the classical counterpart in the processing procedure. Quantum algorithms which benefit from  entanglement show more efficiency on top of their corresponding classical competitors \cite{R.Horodecki&etal:2009}.

However, entanglement is not the only aspect of quantum correlations, and it is shown that some tasks can speed up over their classical counterparts, even in the absence of entanglement. For example the deterministic quantum computation with one qubit (DQC1) uses only separable states \cite{Knill&Laflamme:1998}. In  \cite{Ollivier&Zurek:2001} Ollivier and Zurek have introduced quantum discord as a measure of quantum correlation, beyond entanglement. In addition, Henderson and Vedral have presented measures for total and classical correlations, which are closely related to the quantum discord \cite{Henderson&Vedral:2001}. It is later shown that quantum discord is, indeed,  a quantum correlation resource of DQC1 \cite{Datta&etal:2005,Datta&Vidal:2007,Datta&etal:2008}.

Quantum discord of a bipartite state $\rho$ is defined as the difference between two classically equivalent, but quantum mechanically different definitions of mutual information, i.e.
\begin{equation}
\label{QD}
\mathcal{D}^{(A)}(\rho)=\mathcal{I}(\rho)-\max_{\{\Pi_i^A\}}\mathcal{J}^{(A)}(\rho).
\end{equation}
Here, $\mathcal{I}(\rho)=S(\rho^{A})+S(\rho^{B})-S(\rho)$ is the mutual information of the bipartite state $\rho$, with $S(\cdot)$ being the von Neumann entropy, and $\rho^{A(B)}=\mathrm{Tr}_{B(A)}(\rho)$ is the marginal state corresponding to the subsystem $A(B)$. In addition, $\mathcal{J}^{\{\Pi_i^A\}}(\rho)$ denotes the mutual information of the same system after performing  the local measurement $\{\Pi_i^A\}$ on the subsystem $A$ and can be written as follows \cite{Ollivier&Zurek:2001}
\begin{equation}
\mathcal{J}^{\{\Pi_i^{A}\}}(\rho)=S(\rho^{B})-\sum_ip_iS(\rho|\Pi_i^{A}).
\end{equation}
In this equation $\{\Pi_i^{A}\}=\{|a_i\rangle\langle a_i|\}$ is the set of projection operators on the subsystem $A$, and the conditional state $\rho|\Pi_i^{A}$ is the post-measurement state, i.e.  $\rho|\Pi_i^{A}=\frac{1}{p_i}(\Pi_i^{A}\otimes\mathbb{I}^{B})\rho(\Pi_i^{A}\otimes\mathbb{I}^{B})$, where  $p_i=\mathrm{Tr}[(\Pi_i^{A}\otimes\mathbb{I}^{B})\rho(\Pi_i^{A}\otimes\mathbb{I}^{B})]$ is the probability of the $i$-th outcome of the measurement on the subsystem A \cite{Rulli&Sarandy:2011}. The maximum performed in Eq. \eqref{QD} can be taken also over all the positive operator valued measures (POVM), and these two definitions give, in general, nonequivalent results (see for example \cite{Xu:2011b,Galve&etal:2011,Modi&etal:2012,Shi&etal:2012}).

It would be interesting to mention that the quantum discord, as a measure of quantum correlation, has the following properties \cite{Henderson&Vedral:2001,Modi&etal:2012}: \textit{(i)} $\mathcal{D}^{(A)}(\rho)\ge 0$, where the equality is satisfied  if and only if $\rho$ is classical with respect to the part  $A$, i.e. if and only if $\rho$ is the so-called  classical-quantum state defined by
\begin{equation}
\chi^{A}=\sum_ip_i\Pi_i^{A}\otimes \rho_i^{B},
\label{cqs}
\end{equation}
with $\Pi_i^A=\ket{i}\bra{i}$ as the projection operator on the orthonormal basis of $\mathcal{H}^{A}$, and $\rho_i^B$ being a state on $\mathcal{H}^{B}$.
\textit{(ii)}
$\mathcal{D}^{(A)}(\rho)$ is invariant under local unitary transformations $U_A\otimes U_B$, i.e. $\mathcal{D}^{(A)}((U_A\otimes U_B)\rho(U_A\otimes U_B)^\dagger)=\mathcal{D}^{(A)}(\rho)$.
\textit{(iii)}
Quantum discord reduces to an entanglement monotone for pure states. Indeed $\mathcal{D}^{(A)}(\psi)={\mathcal E}(\psi)$ where ${\mathcal E}(\psi)$ is the entropy of entanglement of the pure state $\ket{\psi}$.
\textit{(iv)} Quantum discord  is not symmetric under the exchange of its two subsystems, i.e. $\mathcal{D}^{(A)}(\rho)$ is not equal to $\mathcal{D}^{(B)}(\rho)$ in general.
\textit{(v)} Quantum discord $\mathcal{D}^{(A)}(\rho)$ can be created by the local channel acting on  part $A$ whose classicality is tested \cite{Gessner&etal:2012}.
\textit{(vi)} Quantum discord $\mathcal{D}^{(A)}(\rho)$ does not increase under the local operations acting on the unmeasured part $B$. For a detailed discussion about the general properties of a measure of classical or quantum correlation, see Refs. \cite{Ollivier&Zurek:2001,Henderson&Vedral:2001,Dakic&etal:2010,Modi&etal:2012,Brodutch&Modi:2012,Gessner&etal:2012,Rulli&Sarandy:2011,Liu&etal:2013,Koashi&Winter:2004}. Specially, in Ref. \cite{Brodutch&Modi:2012} the authors present three types of conditions and call them necessary, reasonable and debatable conditions. Quantum discord satisfies all necessary conditions, besides some reasonable and debatable ones.

Quantum discord is really difficult to be computed, and the analytical solution for this measure has been found only for few special states \cite{Ollivier&Zurek:2001,Henderson&Vedral:2001,Luo:2008,Xu:2013}. This fact makes people look for alternative measures. Among the various measures of quantum correlation, the geometric discord which is first proposed by Dakic  {\etal}, is a simple and intuitive quantifier of the general quantum  correlations  and is defined by \cite{Dakic&etal:2010}
\begin{equation}
\label{GDD}
\mathcal{D}_G^{(A)}(\rho)=\min_{\chi^{A}\in\mathcal{C}^A}\parallel\rho-\chi^{A}\parallel^2_2,
\end{equation}
where $\mathcal{C}^A$ denotes the set of all classical-quantum states defined by Eq. \eqref{cqs}.  In addition, $\|X-Y\|^2_2=\Tr(X-Y)^2$ is the 2-norm or square norm in the Hilbert-Schmidt space. This measure vanishes for the classical-quantum states. Dakic {\etal} also obtain a closed formula for the geometric discord of an arbitrary two-qubit state in terms of its coherence vectors and correlation matrix.
Furthermore, an exact expression for the pure $m\otimes m$ states and arbitrary $2\otimes n$ states are obtained \cite{Luo&Fu:2010,Luo&Fu:2011}.
Luo and Fu show that the definition \eqref{GDD} can also be written as \cite{Luo&Fu:2010}
\begin{equation}
D_G(\rho)=\min_{\Pi^A} \|\rho-\Pi^A(\rho)\|^2_2,
\label{GD-Luo2}
\end{equation}
where the minimum is taken over all von Neumann measurements $\Pi^A=\{\Pi^A_k\}_{k=1}^m$ on ${\mathcal H}^A$, and $\Pi^A(\rho)=\sum_{k=1}^m(\Pi^A_k\otimes \Id)\rho(\Pi^A_k\otimes \Id)$ with $\mathbb{I}$ being the identity operator on the appropriate space. A tight lower bound on the geometric discord of a general bipartite system is given in Refs. \cite{Rana&Parashar:2012,Hassan&etal:2012}, and it is shown that such  lower bound fully captures the quantum correlation of a generic bipartite system, so it  can be used as a measure of quantum correlation in its own right \cite{Akhtarshenas&etal:2015}.

Geometric discord satisfies the properties \textit{(i)-(v)} of the original quantum discord listed above, but as pointed out by Piani \cite{Piani:2012} it fails to satisfy the last property \textit{(vi)}; this measure may increase under local operations on the unmeasured subsystem,  so that one may believe that geometric discord cannot be considered as a good measure of quantum correlations \cite{Piani:2012}.
The lack of contractivity under trace-preserving quantum channels rises from using the Hilbert-Schmidt distance in the definition, and implies that this measure should only be interpreted as a lower bound to an eligible quantum discord measure \cite{Tufarelli&etal:2013}. In addition, the sensitivity of Hilbert-Schmidt norm to the purity of the state, makes the lack of usefulness for this measure even as a lower bond for quantumness in high dimensions \cite{Passante&etal:2012,Brown&etal:2012,Tufarelli&etal:2012}. Note that both problems maybe tackled by choosing a proper norm such as Bures distance or trace distance, but it costs the hardness of calculations to find exact solutions for an arbitrary state. However the latter bug can be solved by removing the sensitivity to the purity of state as is done in Ref. \cite{Tufarelli&etal:2013}, while the simplicity of measure is preserved by using the Hilbert-Schmidt norm yet. The other problem is remedied in Ref. \cite{Paula&etal:2013}, where the Schatten 1-norm is used instead of the Hilbert-Schmidt norm, although the exact solution will not be available for many cases any more. Some other attempts is done in this field, such as using Bures distance in the definition \cite{Spehner&Orszag:2013} or using the square root of density operator rather than the density matrix itself \cite{Chang&Luo:2013}, while these methods make the calculations more complicated too. However geometric discord still can be considered as a good indicator for quantumness of correlations, because of its capability of being computed and its operational significance in specific quantum protocols such as remote state preparation \cite{Dakic&etal:2012}.
Recently, the authors of Ref.  \cite{Akhtarshenas&etal:2015} have introduced the notion of the $A$-correlation matrix and proposed a geometric way of quantifying quantum correlations which can be calculated analytically for the general Schatten $p$-norms and arbitrary functions of the coherence vector of the unmeasured subsystem. They have shown that the measure includes the geometric discord as a special case. Moreover, they have introduced a measure of quantum correlation which is invariant under local quantum channels performing on the unmeasured part, and showed that a way to circumvent the problem with the geometric discord is to re-scale the original geometric discord just by dividing it by the purity of the unmeasured part.

Some attempts have been made in order to generalize the above measures for multipartite systems  \cite{Rulli&Sarandy:2011,Okrasa&Walczak:2011,Hassan&Joag:2012,Xu:2012,Xu:2013}. To generalize a quantum correlation  measure for multipartite states, one method is to symmetrize it with respect to its subsystems. The symmetrization of the quantum discord can be done by measuring, locally, all the subsystems in one step \cite{Rulli&Sarandy:2011}
\begin{equation}
\mathcal{D}(\rho)\equiv\mathcal{D}^{(AB)}(\rho)=\mathcal{I}(\rho)-\max_{\Pi^{AB}}\mathcal{I}\left(\Pi^{AB}(\rho)\right),
\end{equation}
where
\begin{equation}
\Pi^{(AB)}(\rho)=\sum_{i,j}(\Pi_i^A\otimes\Pi_j^B)\rho(\Pi_i^A\otimes\Pi_j^B).
\end{equation}
Some generalization of the symmetric discord to the multipartite states are presented  Refs. \cite{Rulli&Sarandy:2011,Xu:2013}.
For the geometric discord, the generalization is straightforward too and is called the geometric global quantum discord by Xu in Ref. \cite{Xu:2012}
\begin{equation}
\mathcal{D}_{GG}(\rho)=\min_{\chi\in\mathcal{C}}\parallel\rho-\chi\parallel^2_2.
\label{ggqd}
\end{equation}
Here, $\mathcal{C}$ denotes the set of completely classical states which can be written as $\chi=\sum_{i_1,\cdots,i_m}p_{i_1\cdots i_m}\Pi_{i_1}^{A_1}\otimes\cdots\otimes\Pi_{i_m}^{A_m}$, with $\Pi_{i_s}^{A_s}=\project{i_s}$ being the projection operator on the orthonormal basis of the subsystem $A_s$.

The second approach to generalize a measure of  quantum correlation to multipartite systems, is based on the fact that a classical-quantum state can still have quantum correlations which cannot be captured by $\mathcal{D}^{(A)}(\rho)$. This follows from the fact that the states $\rho_i^{B}$, appeared in the definition of classical-quantum states in Eq. \eqref{cqs}, do not commute in general  \cite{Piani&etal:2008,Wu&etal:2009,Okrasa&Walczak:2011}.
It  turns out that the remaining quantum correlations of the state $\rho$ can be obtained  just by calculating $\mathcal{D}^{(B)}(\tilde{\Pi}^A(\rho))$, where $\tilde{\Pi}^A(\rho)$ denotes the post-measurement state of the system after performing the optimized measurement $\tilde{\Pi}^A$ on the part $A$. We therefore obtain the total quantum correlations of the bipartite state $\rho$ as \cite{Okrasa&Walczak:2011,Hassan&Joag:2012}
\begin{equation}
\mathcal{Q}(\rho)=\mathcal{D}^{(A)}(\rho)+\mathcal{D}^{(B)}(\tilde{\Pi}^A(\rho)).
\end{equation}
An extension of this method to the multipartite case is presented in Refs. \cite{Okrasa&Walczak:2011,Hassan&Joag:2012}.

In addition to the different applications of quantum correlations in various areas of quantum information theory,  quantum correlation measures are considered vastly as an important tool to detect the quantum phase transition (QPT) phenomenon \cite{Werlang&etal:2010,Rulli&Sarandy:2011,Cai&etal:2011,Xi&etal:2011,Jie&etal:2012,Sarandy&etal:2013,Campbell&etal:2013}. Quantum discord, as a measure of quantum correlation, can be beneficial in this investigation too. For example, various correlations are studied in the Heisenberg $XXZ$ spin chains, both in the thermal equilibrium and under the intrinsic decoherence, and the QPT has been seen in all of the surveyed correlations, except for the classical one \cite{Cai&etal:2011}. However, all the quantum correlation measures are not always able to detect the critical points. In Ref. \cite{Werlang&etal:2010} it is shown that quantum discord can spotlight the critical points associated to the quantum phase transitions for an infinite chain described by the Heisenberg model, even at finite T, while the entanglement of formation and other thermodynamic quantities cannot. Similar results are observed for a Heisenberg $XXZ$ spin chain after quenches \cite{Jie&etal:2012}. Some other advantages of quantum discord on top of entanglement in the field of detecting QPTs are expressed in Refs. \cite{Xi&etal:2011,Sarandy&etal:2013}.

As spin chains are multipartite states, the generalized quantum correlation measure can be used to study their properties. A symmetrized version of quantum discord, presented in Ref. \cite{Sarandy&etal:2013}, is applied to survey some finite-size spin chains, i.e. transverse field Ising model, cluster-Ising model, and open-chain $XX$ model, and it is shown that, thanks this measure, the critical points can be neatly detected, even for many-body systems that are not in their ground state \cite{Campbell&etal:2013}. As it is shown for Ashkin-Teller spin chain \cite{Rulli&Sarandy:2011}, there are some cases where the pairwise entanglement or quantum discord are not able to detect QPT, while the multipartite quantum discord is capable of detecting it.

In this work,  we present a computable measure of the total quantum correlation of a general multipartite system.  This measure is based on the multipartite extension of the  $A$-correlation matrix of a bipartite state, introduced in Ref.  \cite{Akhtarshenas&etal:2015}. The measure is computable in the sense that it can be computed effectively for an arbitrary mixed state of a multipartite system. We exemplify the measure with some illustrative examples and investigate its ability to detect quantum critical  points in spin chains. The measure therefore can be used as a quantifier for quantum correlations as well as an indicator of the quantumness properties of the multipartite systems.

The remainder of the article is arranged as follows. In  Sec. \ref{s2}, we briefly introduce the measure of quantum correlation introduced recently in Ref. \cite{Akhtarshenas&etal:2015}, and generalize it to capture  total quantum correlation of a bipartite state. In  Sec.  \ref{s3} we generalize the measure for multipartite states. Section \ref{s4} is devoted to show the ability of our measure to detect the quantum critical points in a spin chain. The paper is concluded in Sec.  \ref{s5} with some discussion.

\section{Total quantum correlation: Bipartite systems}
\label{s2}
In this section, we first review the computable measure of the quantum correlation of the bipartite systems and
then extend it to capture all quantum correlations of the bipartite states.

\subsection{Computable measure of quantum correlation}
A general bipartite state $\rho$ on $\mathcal{H}^A\otimes\mathcal{H}^B$ with $\dim(\mathcal{H}^A)=d_A$ and  $\dim(\mathcal{H}^B)=d_B$ can be written as
\begin{eqnarray}
\rho=\frac{1}{d_Ad_B}
\bigg\{\mathbb{I}^A\otimes\mathbb{I}^B&+&\vec{x}\cdot\hat{\lambda}^A\otimes\mathbb{I}^B+\mathbb{I}^A\otimes\vec{y}\cdot\hat{\lambda}^B\nonumber\\
&+&\sum_{i=1}^{d_A^2-1}\sum_{j=1}^{d_B^2-1}t_{ij}\hat{\lambda}_i^A\otimes\hat{\lambda}_j^B\bigg\},
\label{rho_1}
\end{eqnarray}
where $\mathbb{I}^{A}$ and $\Id^{B}$ are the identity matrices on the resoective subspaces $\mathcal{H}^A$ and $\mathcal{H}^B$, while  $\{\hat{\lambda}_i^A\}_{i=1}^{{d_A}^2-1}$ and  $\{\hat{\lambda}_j^B\}_{j=1}^{{d_B}^2-1}$ are generators of $SU(d_A)$ and $SU(d_B)$, respectively, fulfilling the following relations
\begin{eqnarray}\label{SUmGellMann}
 \Tr{\hat{\lambda}_i^s}=0,\qquad \Tr(\hat{\lambda}_i^s\hat{\lambda}_j^s)=2\delta_{ij}, \qquad s=A,B.
\end{eqnarray}
Furthermore, $\vec{x}=(x_1,\cdots,x_{d_A^2-1})^{\T}$ and $\vec{y}=(y_1,\cdots,y_{d_B^2-1})^{\T}$, which are local coherence vectors of the subsystems $A$ and $B$, respectively, and $T=(t_{ij})$, which is called as the correlation matrix, can be written explicitly as follows
\begin{eqnarray}
\label{xyT}
{x}_i&=&\frac{d_A}{2}\mathrm{Tr}\left[(\hat{\lambda}^A_i\otimes\mathbb{I}^B)\rho\right],\nonumber\\
{y}_j&=&\frac{d_B}{2}\mathrm{Tr}\left[(\mathbb{I}^A\otimes\hat{\lambda}^B_j)\rho\right],\\
t_{ij}&=&\frac{d_Ad_B}{4}\mathrm{Tr}\left[(\hat{\lambda}^A_i\otimes\hat{\lambda}^B_j)\rho\right].\nonumber
\end{eqnarray}
Therefore, to each bipartite state $\rho$, one can associate the triple $\{\vec{x},\vec{y},T\}$.
It has been shown that a bipartite state $\rho$ is a classical-quantum state if and only if there exists a $(d_A-1)$-dimensional projection operator $P_A$ on the $(d_A^2-1)$-dimensional space $\mathbb{R}^{d_A^2-1}$ such that \cite{T.Zhou&etal:2011}
\begin{equation}
P_A\vec{x}=\vec{x},\qquad P_AT=T.
\end{equation}
These conditions can be combined as \cite{Akhtarshenas&etal:2015}
\begin{equation}
P_A\mathcal{T}^{A}_{f}=\mathcal{T}^{A}_{f},
\label{nsc1}
\end{equation}
where $\mathcal{T}^A_f$ is the so-called $A$-correlation matrix of the state $\rho$ and is defined as follows
\begin{equation}
\mathcal{T}^{A}_{f}=\sqrt{\frac{2}{d_A^2d_B}}\left(\begin{array}{cc}
f_1({y})\vec{x}\quad&\quad\sqrt{\frac{2}{d_B}}f_2({y})T\end{array}\right).
\end{equation}
Here, $f=\{f_1({y}),f_2({y})\}$, while $f_1({y})$ and $f_2({y})$ are two arbitrary functions of $y=\sqrt{\vec{y}^{\T}\vec{y}}$, with $\vec{y}$ being the local coherence vector of the subsystem $B$. Accordingly, the computable  measure of the $A$-quantum correlation, i.e. quantum correlation of  part $A$,  leads to \cite{Akhtarshenas&etal:2015}
\begin{equation}
\mathcal{Q}_{A}^{f}(\rho)=\min_{P_A}\parallel\mathcal{T}^{A}_f-P_A\mathcal{T}^{A}_f\parallel^2_2=\sum_{k=d_A}^{d_A^2-1}\tau_k^{A,f\downarrow},
\label{ANGDf}
\end{equation}
where $\parallel . \parallel^2_2$ is the square norm in the Hilbert-Schmidt space, $\left\{\tau_k^{A,f\downarrow}\right\}$ are eigenvalues of
$(\mathcal{T}^{A}_f)({\mathcal{T}^A}_f)^{\dagger}=\frac{2}{d_A^2d_B}\left(|f_1(y)|^2 \vec{x}\vec{x}^{\T}+\frac{2}{d_B}|f_2(y)|^2T T^{\T}\right)$
in non-increasing order, and the index $\T$ stands for the transposition. The above measure satisfies all the properties of the original quantum discord if we choose a suitable function for $f$.  Indeed, as it is shown in Ref. \cite{Akhtarshenas&etal:2015}, $\mathcal{Q}_{A}^{f}(\rho)$ satisfies properties \textit{(i),(ii),(iv)} and \textit{(v)}, for a general $f$. Remember that  $\mathcal{Q}_{A}^{f}(\rho)$ is based on the degree to which the necessary and sufficient condition for classicality of part $A$ fails to be satisfied, so  it becomes nonzero for any state obtained by applying suitable local channels \cite{Streltsov&etal:2011b} on a classical state.  Moreover, for any maximally entangled state $\ket{\Psi}=\frac{1}{\sqrt{d_A}}\sum_{i=1}^{d_A}\ket{ii}$,  one can find $\mathcal{Q}_A^f(\rho)=\frac{[d_A(d_A-1)]}{d_A^2}|f_2(0)|^2$, which  achieves its maximum value if $|f_2(y)|$ is a constant or decreasing function of $y=\sqrt{\vec{y}^{\T}\vec{y}}$. As a result, in order to reduce to an entanglement monotone for pure states, i.e. $\mathcal{Q}_{A}^{f}(\rho)$ satisfies the property \textit{(iii)}, $|f_2(y)|$ should be a non-increasing function of $y$. Finally, as it is shown in Ref. \cite{Akhtarshenas&etal:2015}, $\mathcal{Q}_{A}^{f}(\rho)$ may increase under reversible actions on  part $B$, except for the case of choosing  $f_1({y})=f_2({y})=\frac{1}{\sqrt{\mu(\rho^B)}}$, where $\rho^B=\Tr_A(\rho)$ is the reduced density matrix for the subsystem $B$ and $\mu(\cdot)=\Tr[(\cdot)^2]$ is the purity of the state. It turns out that the property \textit{(vi)} is not satisfied, except for this particular case.
To be more specific, let us  concern our attention to the following two interesting choices for $f$.
First, consider $f_1(y)=f_2(y)=1$. Writing   $\mathcal{Q}_{A}(\rho)=\mathcal{Q}_{A}^{f=1}(\rho)$ and $\mathcal{T}^{A}_{f=1}=\mathcal{T}^{A}$ for simplicity, we find
\begin{equation}
\mathcal{Q}_{A}(\rho)=\sum_{k=d_A}^{d_A^2-1}\tau_k^{A\downarrow},
\label{ANGD1}
\end{equation}
where $\left\{\tau_k^{A\downarrow}\right\}$ are the eigenvalues of $(\mathcal{T}^{A})({\mathcal{T}^A})^{\T}=\frac{2}{d_A^2d_B}\left(\vec{x}\vec{x}^{\T}+\frac{2}{d_B}T T^{\T}\right)$ in non-increasing order. As it is shown in Ref. \cite{Akhtarshenas&etal:2015}, in this particular case, the measure is a tight lower bound for the geometric discord \cite{Luo&Fu:2010,Hassan&etal:2012,Rana&Parashar:2012}, and  coincides with it, when the first subsystem is a qubit. Moreover, as mentioned above, in this particular case, all the properties of the original quantum discord, except for the property \textit{(vi)}, are satisfied.

Secondly, let us consider the more interesting case $f_1({y})=f_2({y})=\frac{1}{\sqrt{\mu(\rho^B)}}$, with $\mu(\rho^B)$ as the purity of the reduced state $\rho^B$. Denoting the corresponding $A$-correlation matrix and  $A$-quantum correlation by $\mathcal{T}^{A}_{\mu}$ and $\mathcal{Q}_{A}^{\mu}(\rho)$, respectively, we find
\begin{equation}
\mathcal{Q}_{A}^\mu(\rho)=\min_{P_A}\parallel\mathcal{T}_{\mu}^{A}-P_A\mathcal{T}_{\mu}^{A}\parallel^2_2
=\frac{1}{\mu(\rho^B)}\sum_{k=d_A}^{d_A^2-1}\tau_k^{A\downarrow}.
\label{ANGDmu}
\end{equation}
As it is mentioned above, an important feature of this particular choice for $f$ is that $\mathcal{Q}_{A}^\mu(\rho)$ satisfies all the properties \textit{(i)} to \textit{(vi)} of the original quantum discord. Considering the conditions  presented in Ref. \cite{Brodutch&Modi:2012}, $\mathcal{Q}_{A}^\mu(\rho)$ satisfies all  the necessary conditions for a proper measure of quantum correlation, besides some debatable ones.

Similarly, one can say that the bipartite state $\rho$ is a quantum-classical state, if and only if there exists a $(d_B-1)$-dimensional projection operator $P_B$ on the $(d_B^2-1)$-dimensional space $\mathbb{R}^{d_B^2-1}$ such that
\begin{equation}
P_B\vec{y}=\vec{y}, \quad TP_B=T,
\end{equation}
or equivalently,
\begin{equation}
P_B\mathcal{T}^{B}_{f^\prime}=\mathcal{T}^{B}_{f^\prime},
\label{nsc2}
\end{equation}
where $\mathcal{T}^B_{f^\prime}$ is  the $B$-correlation matrix of the state $\rho$, defined by
\begin{equation}
\mathcal{T}^{B}_{f^\prime}=\sqrt{\frac{2}{d_Ad_B^2}}\left(\begin{array}{cc}
{f^\prime}_1({x})\vec{y}\quad & \quad \sqrt{\frac{2}{d_A}}{f^\prime}_2({x})T^\T\end{array}\right).
\label{bcor}
\end{equation}
with $f^\prime=\{f_1^\prime({x})$, $f_2^\prime({x})\}$ being two arbitrary functions of $x=\sqrt{\vec{x}^{\T}\vec{x}}$, while $\vec{x}$ is the local coherence vector of the subsystem $A$. In a similar manner, we can choose one of the two aforementioned interesting choices for $f^\prime$ and  define $\mathcal{Q}_{B}(\rho)$ and $\mathcal{Q}_{B}^\mu(\rho)$ as counterparts of $\mathcal{Q}_{A}(\rho)$ and $\mathcal{Q}_{A}^\mu(\rho)$, respectively. Letting  $f^\prime_1({x})=f^\prime_2({x})=1$  and writing $\mathcal{T}^{B}_{f^\prime=1}=\mathcal{T}^{B}$, one can define the quantum correlation with respect to the part $B$, as follows
\begin{equation}
\mathcal{Q}_{B}(\rho)=\min_{P_B}\parallel\mathcal{T}^{B}-P_B\mathcal{T}^{B}\parallel^2_2=\sum_{k=d_B}^{d_B^2-1}\tau_k^{B\downarrow},
\label{BNGD}
\end{equation}
where $\left\{\tau_k^{B\downarrow}\right\}$ are the eigenvalues of $(\mathcal{T}^{B})({\mathcal{T}^B})^{\T}$ in non-increasing order.
On the other hand, choosing $f^\prime_1({x})=f^\prime_2({x})=\frac{1}{\sqrt{\mu(\rho^A)}}$, one can define ${Q}_{B}^\mu(\rho)$ as
\begin{equation}
\mathcal{Q}_{B}^\mu(\rho)=\min_{P_B}\parallel\mathcal{T}_{\mu}^{B}-P_B\mathcal{T}_{\mu}^{B}\parallel^2_2
=\frac{1}{\mu(\rho^A)}\sum_{k=d_B}^{d_B^2-1}\tau_k^{B\downarrow}.
\end{equation}
Let us mention here that for the real matrix $\mathcal{T}$,  $\rank(\mathcal{T})=\rank(\mathcal{T}\mathcal{T}^\T)=\rank(\mathcal{T}^\T\mathcal{T})$. Accordingly, condition \eqref{nsc1} implies that a bipartite state $\rho$ is a classical-quantum state, if and only if $\rank({\mathcal{T}^A})\le d_A-1$. Similarly, according to the Eq. \eqref{nsc2}, $\rho$ is a quantum-classical state if and only if $\rank({\mathcal{T}^B})\le d_B-1$.

\subsection{Total quantum correlation of bipartite states}
Let $\rho$ be a bipartite state with associated  triple $\{\vec{x},\vec{y},T\}$, and  $P_A$ and $P_B$ be the $(d_A-1)$- and   $(d_B-1)$-dimensional projection operators acting on the parameter spaces of the first and second subsystems, respectively, namely  on $\mathbb{R}^{d_A^2-1}$ and $\mathbb{R}^{d_B^2-1}$. Let us now perform $P_A$ followed by $P_B$ on the triple $\{\vec{x},\vec{y},T\}$. After performing these projection operators, the triple $\{\vec{x},\vec{y},T\}$ changes, successively, to $\{P_A\vec{x},\vec{y},P_AT\}$ and $\{P_A\vec{x},P_B\vec{y},(P_AT)P_B\}$, associated with the classical-quantum and classical-classical states $\rho{_{{{P}_A}}}$ and $\rho{_{{{P}_B}{{P}_A}}}$, respectively. It follows therefore, that the passage from the initial state $\rho$ to a classical-classical state occurs in two successive steps, corresponding to the local actions $P_A$ followed by $P_B$.  As a result, we can calculate quantum correlations of  the initial and intermediate states $\rho$ and $\rho{_{\tilde{P}_A}}$  corresponding to the respective triples $\{\vec{x},\vec{y},T\}$ and $\{\tilde{P}_A\vec{x},\vec{y},\tilde{P}_AT\}$ as
\begin{eqnarray}
\label{dp1}\mathcal{Q}_{A}(\rho)&=&\min_{P_A}\parallel\mathcal{T}^{A}-P_A\mathcal{T}^{A}\parallel^2_2,\\
\mathcal{Q}_{B}(\rho_{\tilde{P}_A})&=&\min_{P_B}\parallel\mathcal{T}_{\tilde{P}_A}^{B}-P_B\mathcal{T}_{\tilde{P}_A}^{B}\parallel^2_2,
\end{eqnarray}
where $\tilde{P}_A$ is the  projection operator that optimizes Eq. \eqref{dp1}, and $\mathcal{T}_{\tilde{P}_A}^{B}$ is  the $B$-correlation matrix of $\rho_{\tilde{P}_A}$, i.e.
\begin{equation}
\mathcal{T}_{\tilde{P}_A}^{B}=\sqrt{\frac{2}{d_Ad_B^2}}\left(\begin{array}{cc}
\vec{y}\quad & \quad \sqrt{\frac{2}{d_A}}[\tilde{P}_AT]^{\T}\end{array}\right).
\end{equation}
We now define the total $AB$-quantum correlation of $\rho$ as the sum of $A$-quantum correlation of $\rho$ and $B$-quantum correlation of $\rho_{\tilde{P}_A}$, i.e.
\begin{equation}
\mathcal{Q}_{AB}(\rho)=\mathcal{Q}_{A}(\rho)+\mathcal{Q}_{B}(\rho_{\tilde{P}_A}),
\end{equation}
which finally takes the following form after some calculations
\begin{eqnarray}\nonumber
\mathcal{Q}_{AB}(\rho)&=&\min_{P_A}\parallel\mathcal{T}^{A}-P_A\mathcal{T}^{A}\parallel^2_2
+\min_{P_B}\parallel\mathcal{T}_{\tilde{P}_A}^{B}-P_B\mathcal{T}_{\tilde{P}_A}^{B}\parallel^2_2 \\
&=&\sum_{k=d_A}^{d_A^2-1}(\tau^{A\downarrow})_{k}+\sum_{k=d_B}^{d_B^2-1}(\tau_{\tilde{P}_A}^{B\downarrow})_{k},
\label{tqcb}
\end{eqnarray}
where $\left\{(\tau^{A\downarrow})_{k}\right\}$ and $\left\{(\tau^{B\downarrow}_{\tilde{P}_A})_{k}\right\}$ are the eigenvalues of $(\mathcal{T}^{A})({\mathcal{T}^{A}})^{\T}$ and $(\mathcal{T}_{\tilde{P}_A}^{B})({\mathcal{T}_{\tilde{P}_A}^{B}})^{\T}$ in non-increasing order, respectively.
Now, using the notation of the previous section, one may find $\mu[\Tr_A(\rho)]$ and $\mu[\Tr_B(\rho_{\tilde{P}_A})]$ corresponding to the triples $\{\vec{x},\vec{y},T\}$ and $\{\tilde{P}_A\vec{x},\vec{y},\tilde{P}_AT\}$, respectively. As a result
\begin{eqnarray}
\mathcal{Q}_{AB}^\mu(\rho)&=&\mathcal{Q}_{A}^\mu(\rho)+\mathcal{Q}_{B}^\mu(\rho_{\tilde{P}_A})\\
&=&\frac{1}{\mu[\Tr_A(\rho)]}\sum_{k=d_A}^{d_A^2-1}(\tau^{A\downarrow})_{k}+\frac{1}{\mu[\Tr_B(\rho_{\tilde{P}_A})]}\sum_{k=d_B}^{d_B^2-1}(\tau_{\tilde{P}_A}^{B\downarrow})_{k}.\nonumber
\end{eqnarray}

\begin{corollary}
If $d_A<d_B$ then $Q_{AB}=Q_{A}$.
\end{corollary}
\begin{proof}
Note that  for any $(d_A-1)$-dimensional projection operator $P_A$, acting on $\mathbb{R}^{d_A^2-1}$,  $\rank(P_AT)\le d_A-1$, which means that $\rank({\mathcal{T}_{\tilde{P}_A}^{B}})$ is at most $d_A$. However, for a bipartite state $\rho$, $Q_B(\rho)=0$ if and only if $\rank({\mathcal{T}^{B}})\le d_B-1$, implying that $\mathcal{Q}_{B}(\rho_{\tilde{P}_A})=0$  whenever  $d_A\le d_B-1$, or equivalently $d_A<b_B$.
\end{proof}

We can similarly define the total $BA$-quantum correlation of $\rho$ as $\mathcal{Q}_{BA}^f(\rho)=\mathcal{Q}_{B}^f(\rho)+\mathcal{Q}_{A}^f(\rho_{\tilde{P}_B})$ which  is not in general equal to  $\mathcal{Q}_{AB}^f(\rho)=\mathcal{Q}_{A}^f(\rho)+\mathcal{Q}_{B}^f(\rho_{\tilde{P}_A})$ (see the second example below).
This asymmetry is expected from the procedure of our definition of quantum correlation; first of all, we should find the closest  classical-quantum state $\rho_{\tilde{P}_A}$  to the initial state $\rho$ (with respect to the distance defined by Eq. \eqref{ANGDf}), and then find the closest classical-classical state $\rho_{\tilde{P}_B\tilde{P}_A}$ to the obtained classical-quantum state. However, this procedure does not need to be symmetric with respect to the exchange of the two parts $A$ and $B$, as it is illustrated in Fig. \ref{asym}.
It is therefore reasonable to define
\begin{equation}
\mathcal{Q}_{\{AB\}}^f(\rho)=\max\left\{\mathcal{Q}_{AB}^f(\rho),\mathcal{Q}_{BA}^f(\rho)\right\}
\label{symGDA}
\end{equation}
as the total quantum correlation that one can extract from $\rho$. In the above equations, we can use $f=1$ or $f=\mu$ as the proper choices for $f$.
\begin{figure}
\includegraphics[scale=0.8]{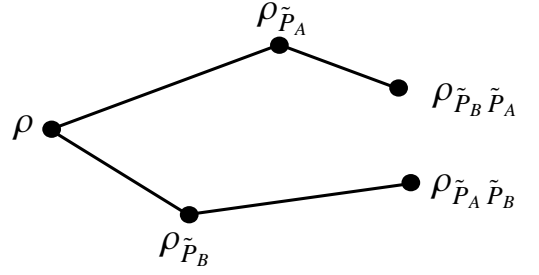}
\caption{A schematic illustration of the two different ways to capture  the total quantum correlation of a given bipartite state $\rho$. In general, neither the total quantum correlations, nor the final classical-classical states are the same.}
\label{asym}
\end{figure}

In order to clarify the defined measure, let us give some illustrative examples.

{\it $m\otimes m$ Werner states---}
A generic $m\otimes m$ Werner states can be defined by
\begin{equation}
\rho_{W}=\frac{m-x}{m^3-m}\mathbb{I}_{m\otimes m}+\frac{mx-1}{m^3-m}F;\qquad x\in[-1,1],
\label{wernerstate}
\end{equation}
where $F=\sum_{k,l=1}^m|kl\rangle\langle lk|$. It is easy to show that
$(\mathcal{T}^A)({\mathcal{T}^A})^{\T}$ has $(m^2-1)$ equal eigenvalues. Moreover,
$\mathcal{Q}_{B}(\rho_{\tilde{P_A}})=\mathcal{Q}_{A}(\rho_{\tilde{P_B}})=0$ and $\mu[\Tr_A(\rho_{W})]=\mu[\Tr_B(\rho_{W})]=\frac{1}{m}$, therefore
\begin{equation}
\mathcal{Q}_{\{AB\}}^{\mu}(\rho)=\mathcal{Q}_{A}^{\mu}(\rho)=\frac{(mx-1)^2}{(m-1)(m+1)^2}\;.
\end{equation}

{\it Mixture of  Bell-diagonal and product states---}
As the second bipartite example, let  $\rho$ be
\begin{equation}
\rho=\lambda\rho_{_{BD}}+(1-\lambda)\rho_p,
\label{asymstate}
\end{equation}
where $0\le \lambda \le 1$, and  $\rho_{_{BD}}$ is the so-called Bell-diagonal state \cite{Ryszard&Michal&Horodecki:1996}
\begin{equation}
\rho_{_{BD}}=\frac{1}{4}(\mathbb{I}\otimes\mathbb{I}+\sum_i t_i\sigma_i\otimes\sigma_i),
\end{equation}
with $t_i$ ($i=1,2,3$) satisfying
\begin{eqnarray}\nonumber
1-t_1-t_2-t_3\ge 0, \qquad 1-t_1+t_2+t_3\ge 0, \\
1+t_1-t_2+t_3\ge 0, \qquad 1+t_1+t_2-t_3\ge 0,
\end{eqnarray}
and $\rho_p$ being a product state
\begin{equation}
\rho_{p}=\frac{\mathbb{I}}{2}\otimes \frac{1}{2}(\mathbb{I}+\vec{r}\cdot\vec{\sigma}),
\end{equation}
with $\sum_{i=1}^3 |r_i|^2\le 1$.  It is straightforward to show that
\begin{equation}
\vec{x}=0, \qquad \vec{y}=(1-\lambda)\vec{r}, \qquad t_{ij}=\lambda t_i\delta_{ij}.
\end{equation}
Constructing $\mathcal{T}^{A}$, we will find
\begin{equation}
(\mathcal{T}^{A})({\mathcal{T}^A})^{\T}=\frac{\lambda^2}{4}\diag\{t_1^2,t_2^2,t_3^2\}.
\end{equation}
Without losing the generality of the problem, we suppose that $|t_1|\geq|t_2|\geq|t_3|$, leading therefore to
\begin{equation}
\mathcal{Q}_{A}(\rho)=\frac{1}{4}\lambda^2\big(t_2^2+t_3^2\big).
\end{equation}
This implies that the optimized projection operator is $\tilde{P}_A=\diag\{1,0,0\}$. Using Eq. \eqref{bcor}, we will find
\begin{equation}
\mathcal{T}_{\tilde{P_A}}^{B}=\frac{1}{2}
\left(\begin{array}{cccc}
r_1(1-\lambda)&t_1\lambda&0&0\\
r_2(1-\lambda)&0&0&0\\
r_3(1-\lambda)&0&0&0
\end{array}\right).
\end{equation}
This leads to
\begin{equation}
Q_B(\rho_{\tilde{P_A}})=\frac{1}{8}\left(h-\sqrt{h^2-k^2}\right),
\end{equation}
where
\begin{eqnarray}
h&=&\lambda^2t_1^2+\left(1-\lambda\right)^2\left(r_1^2+r_2^2+r_3^2\right),\\
k&=&2\lambda\left(1-\lambda\right)t_1\sqrt{r_2^2+r_3^2}.
\end{eqnarray}
Evidently, $Q_B(\rho_{\tilde{P_A}})=0$ if and only if $k=0$.

Now let us evaluate $Q_{BA}(\rho)$ of the above example. In this case we turn  our attention to the simpler case $r_1=r_3=0$ and  obtain
\begin{equation}
\mathcal{Q}_{BA}(\rho)=\frac{1}{4}\big[\lambda^2t_3^2+\min\{\lambda^2t_1^2,\lambda^2t_2^2+r_2^2(1-\lambda)^2\}\big].
\label{rl}
\end{equation}
On the other hand, $Q_{AB}(\rho)$ for this particular case becomes
\begin{equation}
\mathcal{Q}_{AB}(\rho)=\frac{1}{4}\big[\lambda^2(t_2^2+t_3^2)+\min\{\lambda^2t_1^2,r_2^2(1-\lambda)^2\}\big].
\label{lr}
\end{equation}
A comparison between  Eq. \eqref{rl} and Eq. \eqref{lr} implies that these two expressions for the total quantum correlation are not equivalent. Fig. \ref{asym3D}(a) illustrates this asymmetry for the special case $\lambda=\frac{1}{2}$ and $t_1=t_2=t_3=\frac{1}{3}(2x-1)$, i.e. $\rho$ being a uniform mixture of the Werner state (with the parameter $x$ as in Eq. \eqref{wernerstate}) and the product state $\rho_p$. This figure  reveals that, for these values of parameters, we have $\mathcal{Q}_{AB}(\rho)\ge \mathcal{Q}_{BA}(\rho)$, hence $\mathcal{Q}_{\{AB\}}(\rho)=\mathcal{Q}_{AB}(\rho)$. It is easy to show that for this case, $\mu[Tr_A(\rho)]=\frac{1}{2}(1-\frac{r_2^2}{4})$, $\mu[Tr_B(\rho)]=\frac{1}{2}$ and $\mu[Tr_B(\rho_{\tilde{P}_A})]=\frac{1}{2}$, so that
\begin{equation}
\mathcal{Q}_{AB}^{\mu}(\rho)=\frac{1}{\mu[Tr_A(\rho)]}\mathcal{Q}_{A}(\rho)+\frac{1}{\mu[Tr_B(\rho_{\tilde{P}_A})]}\mathcal{Q}_{B}(\rho_{\tilde{P}_A}),
\end{equation}
and as $\mathcal{Q}_{A}(\rho_{\tilde{P}_B})=0$, we find
\begin{equation}
\mathcal{Q}_{BA}^{\mu}(\rho)=\frac{1}{\mu[Tr_B(\rho)]}\mathcal{Q}_{BA}(\rho).
\end{equation}
We have illustrated $\mathcal{Q}_{AB}^{\mu}(\rho)$ and $\mathcal{Q}_{BA}^{\mu}(\rho)$ in Fig. \ref{asym3D}(b).

\begin{figure}
\mbox{\includegraphics[scale=0.4]{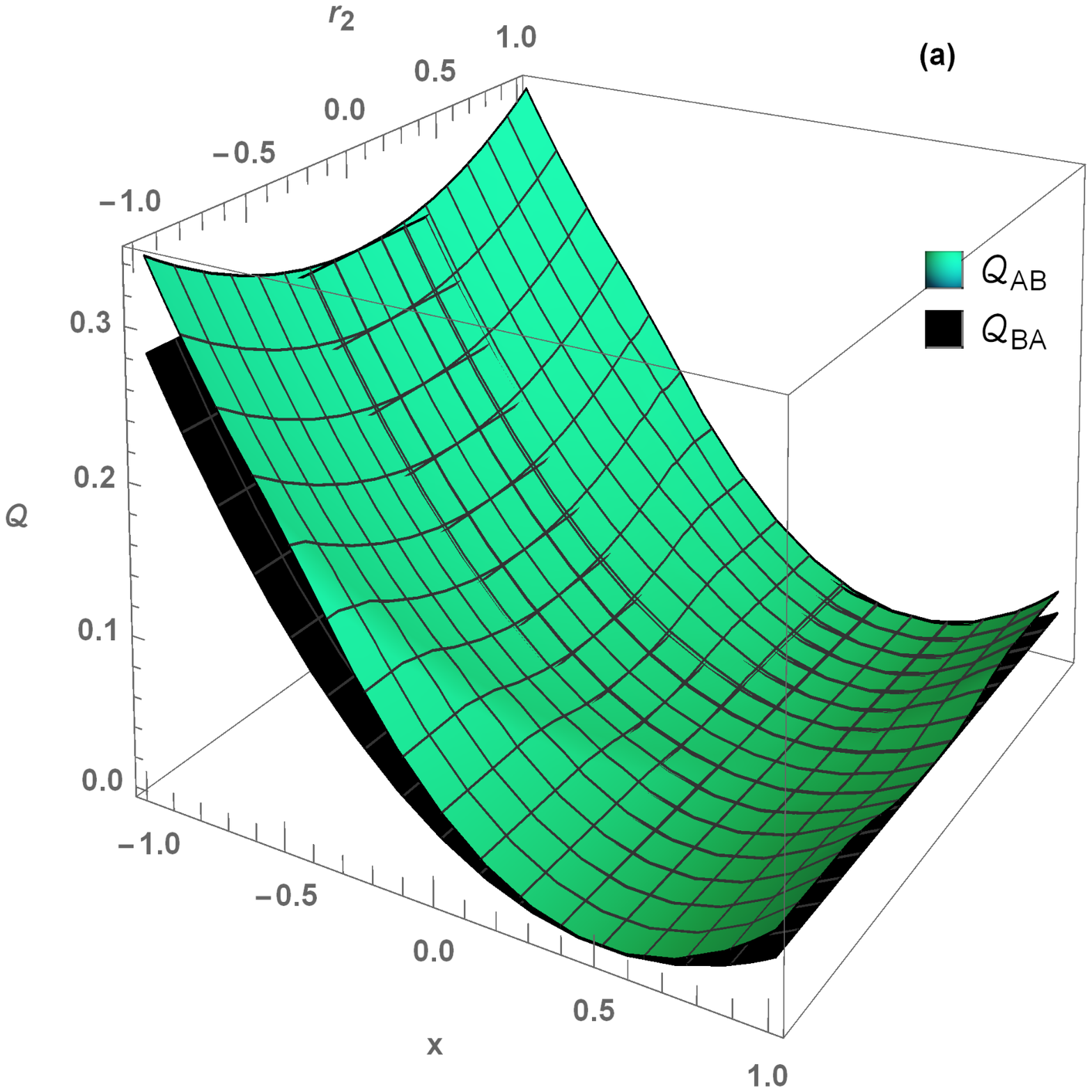}}
\mbox{\includegraphics[scale=0.4]{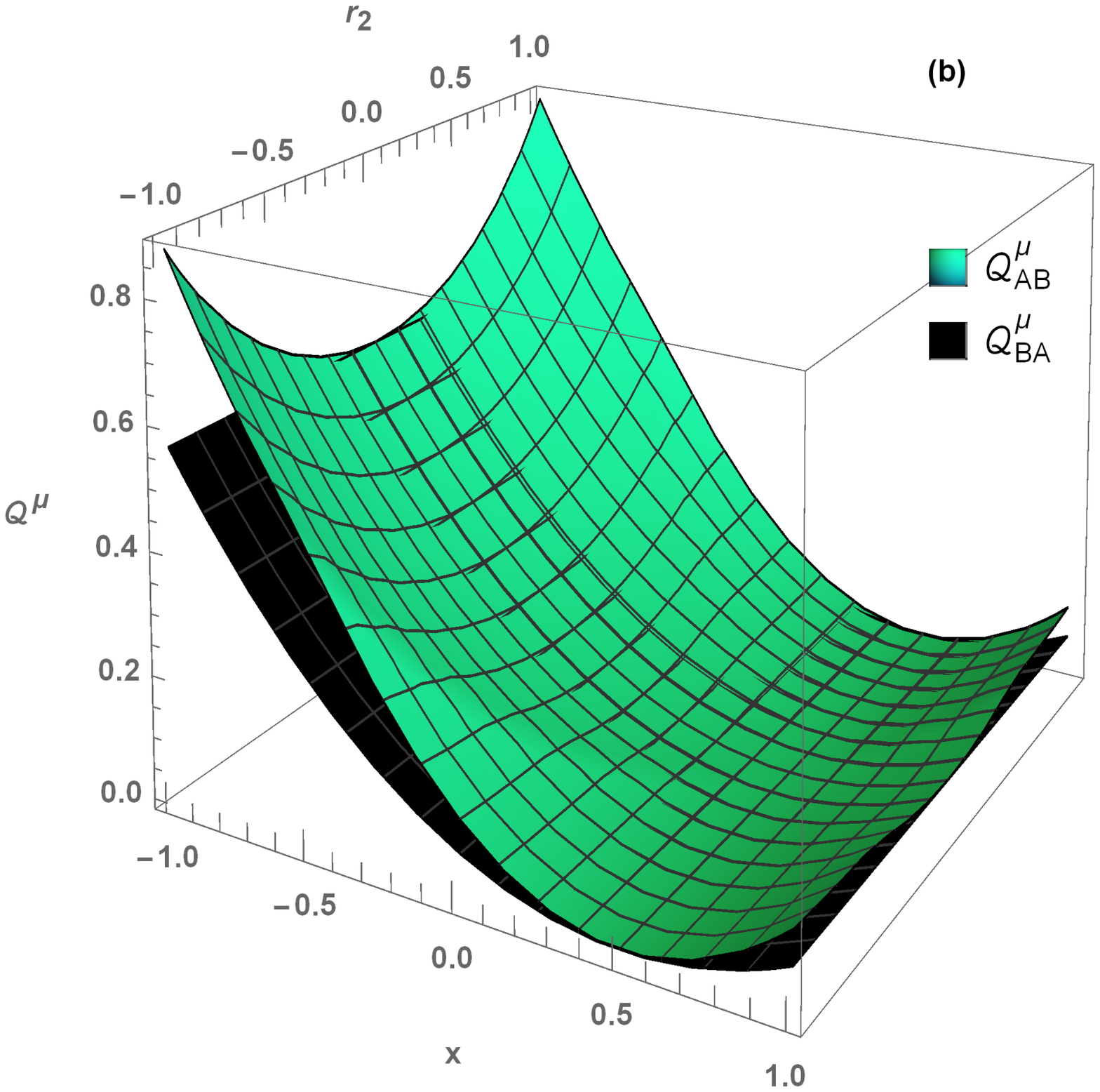}\label{asymtilde3D}}
\caption{(Color online) Total quantum correlations (a) $\mathcal{Q}_{AB}(\rho)$ and $\mathcal{Q}_{BA}(\rho)$ (b) $\mathcal{Q}_{AB}^{\mu}(\rho)$ and $\mathcal{Q}_{BA}^{\mu}(\rho)$ of the state given by  Eq. \eqref{asymstate} as  functions of $x$ and $r_2$. Here we choose $t_1=t_2=t_3=\frac{1}{3}(2x-1)$ , $r_1=r_3=0$ and $\lambda=\frac{1}{2}$. }
\label{asym3D}
\end{figure}

\section{Total quantum correlation: Multipartite systems}

\label{s3}
The generalization of the measure defined in the previous section for multipartite systems, seems to be  a little difficult. The origin of this obstacle rises from the fact that, instead of the correlation matrix $T$ appeared in the description of a bipartite state,  we now encounter with tensors  of $\rank$s larger than two. As we will show below, one can simply overcome this ambiguity by generalizing the notions of $A$- and $B$- correlation matrices.
To begin with, let us first introduce  an alternative extension for a multipartite state.  Let $\{X_{i}^{(A_s)}\}_{i=0}^{d_{A_s}^2-1}$ (for $s=1,2,\cdots,m$)
\begin{equation}
X_0^{(A_s)}=\frac{1}{\sqrt{d_{A_s}}}\mathbb{I}^{(A_s)},\quad X_{i\neq0}^{(A_s)}=\frac{1}{\sqrt{2}}\hat{\lambda}_i^{(A_s)},
\end{equation}
be  the set of Hermitian operators which constitute  an orthonormal basis for the space $\mathcal{B}(\mathcal{H}^{A_s})$, i.e.  $\mathrm{Tr}(X_{i_s}^{(A_s)}X_{j_s}^{(A_s)})=\delta_{i_sj_s}$. The above  $\hat{\lambda}_i^{(A_s)}$'s denote the set of generalized Gell-Mann $d_{A_s}\times d_{A_s}$ matrices.
Then a general $m$-partite state $\rho$ on $\mathcal{H}=\prod_{s=1}^{m}\otimes\mathcal{H}^{A_s}$ with $\dim(\mathcal{H}^{A_s})=d_{A_s}$ ($s=1,\cdots,m$)  can be written as
\begin{equation}
\rho=\sum_{i_1=0}^{d_{A_1}^2-1}\cdots\sum_{i_m=0}^{d_{A_m}^2-1}C_{i_1\cdots i_m}X_{i_1}^{(A_1)}\otimes\cdots\otimes X_{i_m}^{(A_m)},
\end{equation}
where
\begin{equation}
\label{cij}
C_{i_1\cdots i_m}=\mathrm{Tr}\left[(X_{i_1}^{(A_1)}\otimes\cdots\otimes X_{i_m}^{(A_m)})\rho\right].
\end{equation}
Comparing with Eq. \eqref{rho_1} in the bipartite case, we find  $x_i=\sqrt{\frac{d_A^2d_B}{2}}C_{i0}$, $y_j=\sqrt{\frac{d_Ad_B^2}{2}}C_{0j}$ and $t_{ij}=\frac{1}{2}d_Ad_BC_{ij}$.
Accordingly, a general state of a generic $m$-partite system can be represented by a $M$-\textit{tuple} with $M=\sum_{k=1}^m M_k$ where $M_k=\frac{m!}{k!(m-k)!}$. The entries of such $M$-\textit{tuple} are constructed by $M_k$ tensors of rank $k$ with  $k=1,2,\cdots,m$.

\subsection{Tripartite case}
To be more specific, let us turn our attention to the particular case $m=3$, i.e. tripartite systems. In this case, we can define three tensors of $rank$ 1, three tensors of $rank$ 2, and one tensor of $rank$ 3, as follows
\begin{eqnarray}
&&V_{i_1}^{(A_1)}=C_{i_100}\hspace{6.0mm},\hspace{3.0mm}V_{i_2}^{(A_2)}=C_{0i_20}\hspace{3.8mm},\hspace{6.0mm}V_{i_3}^{(A_3)}=C_{00i_3},\nonumber\\
&&T_{i_1i_2}^{(A_1A_2)}=C_{i_1i_2 0}\hspace{1mm},\hspace{1mm}T_{i_2i_3}^{(A_2A_3)}=C_{0i_2i_3}\hspace{1mm},\hspace{1mm}T_{i_1i_3}^{(A_1A_3)}=C_{i_1 0 i_3},\nonumber\\
&&T_{i_1i_2i_3}^{(A_1A_2A_3)}=C_{i_1i_2i_3}.
\end{eqnarray}
However, our task is to represent all coefficients $C_{i_1i_2i_3}$  in the matrix form. Toward this aim, we define  three collections  of matrices as $\left\{T_{[i_3]}^{(A_1A_2)}\right\}_{i_3=0}^{d_{A_3}^2-1}$, $\left\{T_{[i_2]}^{(A_1A_3)}\right\}_{i_2=0}^{d_{A_2}^2-1}$, and $\left\{T_{[i_1]}^{(A_2A_3)}\right\}_{i_1=0}^{d_{A_1}^2-1}$ with the following matrix elements
\begin{eqnarray}
T_{i_1i_2[i_3]}^{(A_1A_2)}&=&C_{i_1i_2i_3},\quad i_3=0,\cdots, d_{A_3}^2-1,\\
T_{i_1[i_2]i_3}^{(A_1A_3)}&=&C_{i_1i_2i_3},\quad i_2=0,\cdots, d_{A_2}^2-1,\\
T_{[i_1]i_2i_3}^{(A_2A_3)}&=&C_{i_1i_2i_3},\quad i_1=0,\cdots, d_{A_1}^2-1.
\end{eqnarray}
Having these constitution matrices in hand, we can now construct  $A_1$-correlation matrix $\mathcal{T}^{(A_1)}$ as
\begin{equation}
\mathcal{T}^{(A_1)}=\left\{\vec{V}^{(A_1)},  \left[T_{i_1i_2[i_3]}^{(A_1A_2)}\right], \left[T_{i_1[i_2=0]i_3}^{(A_1A_3)}\right]  \right\},
\label{a1cor}
\end{equation}
where in the second term, the index $i_3$ ranges over its  appropriate domain. Note that fixing the other index, $i_2=0$, in the third term is to prevent overcounting.  Similarly, one can define $A_2$- and $A_3$- correlation matrices as
\begin{equation}
\mathcal{T}^{(A_2)}=\left\{\vec{V}^{(A_2)},  \left[T_{i_1i_2[i_3]}^{(A_1A_2)}\right]^{\T}, \left[T_{[i_1=0]i_2i_3}^{(A_2A_3)}\right]  \right\},
\label{a2cor}
\end{equation}
\begin{equation}
\mathcal{T}^{(A_3)}=\left\{\vec{V}^{(A_3)},  \left[T_{i_1[i_2]i_3}^{(A_1A_3)}\right]^{\T}, \left[T_{[i_1=0]i_2i_3}^{(A_2A_3)}\right]^{\T}  \right\}.
\label{a3cor}
\end{equation}
 Note the transpose symbol "$\T$" in the appropriate places. Indeed, when we construct the $A_s$-correlation matrix $\mathcal{T}^{(A_s)}$ (for $s=1,2,3$), the corresponding constituted matrices are transposed when the index $i_s$ appears as the column index.
Using these definitions, we end up with the following observation for quantum correlation due to party  $A_s$ ($s=1,2,3$).
\begin{observation}\label{Obs}
The tripartite state $\rho$ is classical with respect to the part $A_s$ ($s=1,2,3$), if and only if there exists a $(d_{A_s}-1)$-dimensional projection operator $P_{A_s}$ acting  on $\mathbb{R}^{d_{A_s}^2-1}$ such that
\begin{equation}
P_{A_s}\mathcal{T}^{(A_s)}=\mathcal{T}^{(A_s)},
\end{equation}
where $\mathcal{T}^{(A_s)}$ is defined as above.
\end{observation}
\begin{proof} We prove the observation for the first subsystem, while generalizing to the other subsystems is straightforward. Let $\rho$ be a tripartite state acting on $\mathcal{H}=\mathcal{H}^{A_1}\otimes\mathcal{H}^{A_2}\otimes\mathcal{H}^{A_3}$, and $C_{i_1i_2i_3}$ be expansion coefficients of $\rho$ in terms of the orthonormal basis $X_{i_1}^{(A_1)}\otimes X_{i_2}^{(A_2)}\otimes X_{i_3}^{(A_3)}$. However, one can partition $\mathcal{H}$ as $\mathcal{H}=\mathcal{H}^{A}\otimes\mathcal{H}^{B}$ with $A=A_1$ and $B=(A_2A_3)$ as the new subsystems.
Applying this division, the coherence vector $\vec{x}$ and the correlation matrix $T$ of $\rho$, as a bipartite state, are given respectively by
$x_{i_1}=\sqrt{\frac{d_A^2d_B}{2}}C_{i_100}$, and $t_{i_1(i_2i_3)}=\frac{1}{2}d_Ad_BC_{i_1i_2i_3}$ with  $d_A=d_{A_1}$, and $d_B=d_{A_2}d_{A_3}$. Here we have used the index string $(i_2i_3)$ as a collective index for the columns of $T$.
This means that the bipartite correlation matrix $T=(t_{i_1(i_2i_3)})$ is constructed as $T=\left[\left[T_{i_1i_2[i_3]}^{(A_1A_2)}\right], \left[T_{i_1[i_2=0]i_3}^{(A_1A_3)}\right]\right]$ where in the first term $i_3$ ranges over its appropriate domain, while in the second one, $i_2$ is fixed to zero for the purpose of preventing overcounting.
It follows therefore that $\rho$, as a bipartite state,  is a classical state with respect to part $A$ if and only if the associated bipartite $A$-correlation matrix $\mathcal{T}^{A}=\sqrt{\frac{2}{d_A^2d_B}}\left(\begin{array}{cc}
\vec{x}\quad & \quad \sqrt{\frac{2}{d_B}}T\end{array}\right)$ satisfies condition \eqref{nsc1}. One can find that the $A_1$-correlation matrix of the tripartite state $\rho$ is, indeed, nothing but the the associated bipartite $A$-correlation matrix, so the proof is complete.
\end{proof}

This condition allows us to define the following measure for quantum correlation of part $A_s$, which is a multipartite generalization of the measure defined in Ref. \cite{Akhtarshenas&etal:2015}
\begin{equation}
\mathcal{Q}_{A_s}(\rho)=\min_{P_{A_s}}\parallel\mathcal{T}^{(A_s)}-P_{A_s}\mathcal{T}^{(A_s)}\parallel^2_2.
\label{GGQDs}
\end{equation}
Now let $\tilde{P}_{A_1}$ be the projection operator that minimizes the above distance for $s=1$, i.e.
\begin{equation}
\mathcal{Q}_{A_1}(\rho)=\parallel\mathcal{T}^{(A_1)}-\tilde{P}_{A_1}\mathcal{T}^{(A_1)}\parallel^2_2.
\end{equation}
After performing this projection operator, the $A_s$-correlation matrices of $\rho$, i.e.  $\mathcal{T}^{(A_s)}$ (for $s=1,2,3$), transform  to
\begin{eqnarray}
\mathcal{T}_{\tilde{P}_{A_1}}^{(A_1)}&=&\left\{\tilde{P}_{A_1}\vec{V}^{(A_1)}, \left[T_{i_1i_2[i_3]}^{(A_1A_2)}\right]_{\tilde{P}_{A_1}}, \left[T_{i_1[i_2=0]i_3}^{(A_1A_3)}\right]_{\tilde{P}_{A_1}}  \right\},\nonumber\\
\mathcal{T}_{\tilde{P}_{A_1}}^{(A_2)}&=&\left\{\vec{V}^{(A_2)}, \left[T_{i_1i_2[i_3]}^{(A_1A_2)}\right]_{\tilde{P}_{A_1}}^{\T}, \left[T_{[i_1=0]i_2i_3}^{(A_2A_3)}\right]_{\tilde{P}_{A_1}}  \right\},\nonumber\\
\mathcal{T}_{\tilde{P}_{A_1}}^{(A_3)}&=&\left\{\vec{V}^{(A_3)}, \left[T_{i_1[i_2]i_3}^{(A_1A_3)}\right]_{\tilde{P}_{A_1}}^{\T}, \left[T_{[i_1=0]i_2i_3}^{(A_2A_3)}\right]_{\tilde{P}_{A_1}}^{\T}  \right\}.
\end{eqnarray}
These matrices can be regarded as the $A_s$-correlation matrices of the state $\rho_{\tilde{P}_{A_1}}$ wherein
\begin{eqnarray}
\left[T_{i_1i_2[i_3]}^{(A_1A_2)}\right]_{\tilde{P}_{A_1}}=
\sum_{i^{\prime}_1=1}^{d_{A_1}^2-1}\left(\tilde{P}_{A_1}\right)_{i_1i_1^{\prime}}T_{i_1^{\prime}i_2[i_3]}^{(A_1A_2)},\\
\left[T_{i_1[i_2]i_3}^{(A_1A_3)}\right]_{\tilde{P}_{A_1}}=
\sum_{i^{\prime}_1=1}^{d_{A_1}^2-1}\left(\tilde{P}_{A_1}\right)_{i_1i_1^{\prime}}T_{i_1^{\prime}[i_2]i_3}^{(A_1A_3)},\\
\left[T_{[i_1]i_2i_3}^{(A_2A_3)}\right]_{\tilde{P}_{A_1}}=
\sum_{i^{\prime}_1=1}^{d_{A_1}^2-1}\left(\tilde{P}_{A_1}\right)_{i_1i_1^{\prime}}T_{[i_1^{\prime}]i_2i_3}^{(A_2A_3)},
\end{eqnarray}
for $1\leq i_1\leq d_{A_1}^2-1$, $0\leq i_2\leq d_{A_2}^2-1$, and $0\leq i_3\leq d_{A_3}^2-1$.
However, for $i_1=0$ the trivial relations will be preserved, i.e. we have $\left[T_{0i_2[i_3]}^{(A_1A_2)}\right]_{\tilde{P}_{A_1}}=\left[T_{0i_2[i_3]}^{(A_1A_2)}\right]$, $\left[T_{0[i_2]i_3}^{(A_1A_3)}\right]_{\tilde{P}_{A_1}}=\left[T_{0[i_2]i_3}^{(A_1A_3)}\right]$, and $\left[T_{[0]i_2i_3}^{(A_2A_3)}\right]_{\tilde{P}_{A_1}}=\left[T_{[0]i_2i_3}^{(A_2A_3)}\right]$.

Clearly, $\rho_{\tilde{P}_{A_1}}$ has zero $A_1$-quantum correlation, but it may still have nonzero $A_2$- and $A_3$- quantum correlations. Now suppose $\tilde{P}_{A_2}$ be the  projection operator that gives the $A_2$-quantum correlation of $\rho_{\tilde{P}_{A_1}}$, i.e.
\begin{equation}
\mathcal{Q}_{A_2}(\rho_{\tilde{P}_{A_1}})=\parallel\mathcal{T}_{\tilde{P}_{A_1}}^{(A_2)}-\tilde{P}_{A_2}\mathcal{T}_{\tilde{P}_{A_1}}^{(A_2)}\parallel^2_2.
\end{equation}
Performing $\tilde{P}_{A_2}$ changes the $A_3$-correlation matrix $\mathcal{T}_{\tilde{P}_{A_1}}^{(A_3)}$  as follows
\begin{equation}
\mathcal{T}_{\tilde{P}_{A_2}\tilde{P}_{A_1}}^{(A_3)}=\left\{\vec{V}^{(A_3)}, \left[T_{i_1[i_2]i_3}^{(A_1A_3)}\right]_{\tilde{P}_{A_2}\tilde{P}_{A_1}}^{\T}, \left[T_{[i_1=0]i_2i_3}^{(A_2A_3)}\right]_{\tilde{P}_{A_2}\tilde{P}_{A_1}}^{\T}  \right\},
\end{equation}
where
\begin{eqnarray}
\left[T_{i_1[i_2]i_3}^{(A_1A_3)}\right]_{\tilde{P}_{A_2}\tilde{P}_{A_1}}=\sum_{i_2^{\prime}=1}^{d_{A_2}^2-1}\left(\tilde{P}_{A_2}\right)_{i_2i_2^{\prime}}\left[T_{i_1[i_2^{\prime}]i_3}^{(A_1A_3)}\right]_{\tilde{P}_{A_1}}\\
\left[T_{[i_1]i_2i_3}^{(A_2A_3)}\right]_{\tilde{P}_{A_2}\tilde{P}_{A_1}}=\sum_{i_2^{\prime}=1}^{d_{A_2}^2-1}\left(\tilde{P}_{A_2}\right)_{i_2i_2^{\prime}}\left[T_{[i_1]i_2^{\prime}i_3}^{(A_2A_3)}\right]_{\tilde{P}_{A_1}}
\end{eqnarray}
for $0\leq i_1\leq d_{A_1}^2-1$, $1\leq i_2\leq d_{A_2}^2-1$ and $0\leq i_3\leq d_{A_3}^2-1$.
Again, note that for $i_2=0$ the trivial relations will be preserved, i.e. $\left[T_{i_1[0]i_3}^{(A_1A_3)}\right]_{\tilde{P}_{A_2}\tilde{P}_{A_1}}=\left[T_{i_1[0]i_3}^{(A_1A_3)}\right]_{\tilde{P}_{A_1}}$ and $\left[T_{[i_1]0i_3}^{(A_2A_3)}\right]_{\tilde{P}_{A_2}\tilde{P}_{A_1}}=\left[T_{[i_1]0i_3}^{(A_2A_3)}\right]_{\tilde{P}_{A_1}}$.

Finally, we can define the $A_3$-quantum correlation of $\rho_{\tilde{P}_{A_2}\tilde{P}_{A_1}}$ as follows
\begin{equation}
\mathcal{Q}_{A_3}(\rho_{\tilde{P}_{A_2}\tilde{P}_{A_1}})=\min_{P_{A_3}}\parallel\mathcal{T}_{\tilde{P}_{A_2}\tilde{P}_{A_1}}^{(A_3)}
-P_{A_3}\mathcal{T}_{\tilde{P}_{A_2}\tilde{P}_{A_1}}^{(A_3)}\parallel^2_2.
\end{equation}
We define therefore the total $A_1A_2A_3$-quantum correlation of the tripartite state $\rho$ as
\begin{equation}
\mathcal{Q}_{A_1A_2A_3}(\rho)=\mathcal{Q}_{A_1}(\rho)+\mathcal{Q}_{A_2}(\rho_{\tilde{P}_{A_1}})+\mathcal{Q}_{A_3}(\rho_{\tilde{P}_{A_2}\tilde{P}_{A_1}}),
\label{tq3}
\end{equation}
which can be expressed as
\begin{eqnarray}
\mathcal{Q}_{A_1A_2A_3}(\rho)&=&\sum_{k=d_{A_1}}^{d_{A_1}^2-1}(\tau^{(A_1)\downarrow})_k
+\sum_{k=d_{A_2}}^{d_{A_2}^2-1}(\tau^{(A_2)\downarrow}_{\tilde{P}_1})_{k}\nonumber \\
&&+\sum_{k=d_{A_3}}^{d_{A_3}^2-1}(\tau^{(A_3)\downarrow}_{\tilde{P}_2\tilde{P}_1})_{k}
\end{eqnarray}
where $\left\{(\tau^{(A_1)\downarrow})_{k}\right\}$, $\left\{(\tau^{(A_2)\downarrow}_{\tilde{P}_1})_{k}\right\}$ and $\left\{(\tau^{(A_3)\downarrow}_{\tilde{P}_2\tilde{P}_1})_{k}\right\}$ are respective eigenvalues of $(\mathcal{T}^{(A_1)})({\mathcal{T}^{(A_1)}})^{\T}$, $(\mathcal{T}_{\tilde{P_1}}^{(A_2)})({\mathcal{T}_{\tilde{P_1}}^{(A_2)}})^{\T}$, and $(\mathcal{T}_{\tilde{P_2}\tilde{P_1}}^{(A_3)})({\mathcal{T}_{\tilde{P_2}\tilde{P_1}}^{(A_3)}})^{\T}$ in non-increasing order. Finally, we define the total quantum correlation of the tripartite state $\rho$ as
\begin{equation}
\mathcal{Q}_{\{A_1A_2A_3\}}=\max_{\mathcal{P}}\{\mathcal{Q}_{A_{i_1}A_{i_2}A_{i_3}}\}
\end{equation}
where maximum is taken over  all the $3!$ permutations of the three subsystems $\{A_1A_2A_3\}$.

 In a similar manner we define
\begin{equation}
\mathcal{Q}^\mu_{\{A_1A_2A_3\}}=\max_{\mathcal{P}}\{\mathcal{Q}^{\mu}_{A_{i_1}A_{i_2}A_{i_3}}\}
\label{QT3P}
\end{equation}
where, for example, $\mathcal{Q}^\mu_{A_1A_2A_3}(\rho)$ is given by
\begin{eqnarray}
\mathcal{Q}^\mu_{A_1A_2A_3}(\rho)&=&\mathcal{Q}^\mu_{A_1}(\rho)+\mathcal{Q}^\mu_{A_2}(\rho_{\tilde{P}_{A_1}})+\mathcal{Q}^\mu_{A_3}(\rho_{\tilde{P}_{A_1A_2}})\nonumber\\
&=&\frac{1}{\mu[\Tr_{A_1}(\rho)]}\mathcal{Q}_{A_1}(\rho)+\frac{1}{\mu[\Tr_{A_2}(\rho_{\tilde{P}_{A_1}})]}\mathcal{Q}_{A_1}(\rho_{\tilde{P}_{A_1}})\nonumber\\
&&+\frac{1}{\mu[\Tr_{A_3}(\rho_{\tilde{P}_{A_2}\tilde{P}_{A_1}})]}\mathcal{Q}_{A_3}(\rho_{\tilde{P}_{A_2}\tilde{P}_{A_1}}),
\label{QTT3}
\end{eqnarray}
Let consider some illustrative examples.

{\it Werner-GHZ state}
A generic  Werner-GHZ state is defined as
\begin{equation}
\rho=\frac{1-\lambda}{8}\mathbb{I}^{\otimes 3}+\lambda|GHZ\rangle\langle GHZ|;\quad0\leq \lambda\leq1,
\end{equation}
with $|GHZ\rangle=\frac{1}{\sqrt{2}}\left(|000\rangle+|111\rangle\right)$. Following the above procedure, one can easily show that $\mathcal{Q}_{A_1}(\rho)=\frac{1}{2}\lambda^2$, $\mathcal{Q}_{A_2}(\rho_{\tilde{P}_{A_1}})=\frac{1}{4}\lambda^2$, and $\mathcal{Q}_{A_3}(\rho_{\tilde{P}_{A_2}\tilde{P}_{A_1}})=0$.  As the state is symmetric with respect to its three subsystems, the total quantum correlation of this state is
\begin{equation}
\mathcal{Q}_{\{A_1A_2A_3\}}=\mathcal{Q}_{A_1A_2A_3}=\frac{3}{4}\lambda^2.
\end{equation}
Using  the fact that $\mu(\Tr_A[\rho])=\frac{1}{4}(1+\lambda^2)$ and $\mu(\Tr_B[\rho_{\tilde{P}_A}])=\frac{1}{4}$ we  find
\begin{equation}
\mathcal{Q}_{\{A_1A_2A_3\}}^{\mu}=\mathcal{Q}_{A_1A_2A_3}^{\mu}=\frac{\lambda^2(3+\lambda^2)}{1+\lambda^2}.
\end{equation}
The symmetric quantum discord  and  the geometric global quantum discord  of this state are also calculated in \cite{Rulli&Sarandy:2011} and \cite{Xu:2012} as
\begin{eqnarray}
\mathcal{D}_{s}(\rho)&=&-\frac{1}{4}(1+3\lambda)\log_2(1+3\lambda)+\frac{1}{8}(1-\lambda)\log_2(1-\lambda)\nonumber\\
&&+\frac{1}{8}(1+7\lambda)\log_2(1+7\lambda),
\end{eqnarray}
and
\begin{equation}
\mathcal{D}_{GG}(\rho)=\frac{1}{2}\lambda^2,
\end{equation}
respectively. For a comparison, we have plotted the results in Fig. \ref{WernerGHZ}.
\begin{figure}
\includegraphics[scale=0.8]{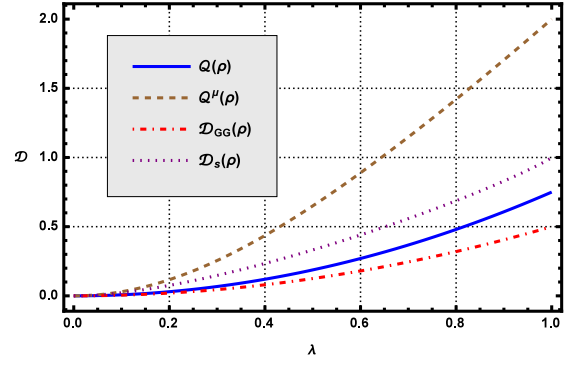}
\caption{(Color online) Geometric global quantum discord, generalized symmetric quantum discord, and our computable measure of total quantum correlation of the Werner-GHZ state as functions of $\lambda$.}
\label{WernerGHZ}
\end{figure}

{\it W-GHZ state---} A generic W-GHZ state has the following definition \cite{Hassan&Joag:2012}
\begin{equation}
\rho_2=\lambda|W\rangle\langle W|+(1-\lambda)|GHZ\rangle\langle GHZ|,
\end{equation}
where $|W\rangle=\frac{1}{\sqrt{3}}\left(|100\rangle+|010\rangle+|001\rangle\right)$,  and $\lambda$ ranges from 0 to 1. Although the analytical calculations can be made for this example, because of their complicated form we only illustrate the results in Fig. \ref{W-GHZ}. Because of the behavior of the eigenvalues of $(\mathcal{T}^{(A_1)})({\mathcal{T}^{(A_1)}})^{\T}$, the $A_1$-quantum correlation  for this state is a two-conditional function.  As a result, we will have two different projectors to gain the quantum correlation, making  the total quantum correlation discontinuous.  Due to the change in the behavior of the eigenvalues of $A_1$-correlation matrix at $\lambda_0=\frac{3}{4}$ \cite{Hassan&Joag:2012}, the sudden change \cite{Maziero&etal:2009,Mazzola&etal:2010,Wu&etal:2011,Joao&etal:2013,Jia&etal:2013} in the geometric discord occurs, leading therefore to the discontinuity in the total quantum correlation. More precisely, for $\lambda<\lambda_0$ and $\lambda>\lambda_0$, the projection operator $\tilde{P}_{A_1}$ leads to two different $\rho_{\tilde{P}_{A_1}}$ with different quantum correlations.
\begin{figure}
\includegraphics[scale=0.8]{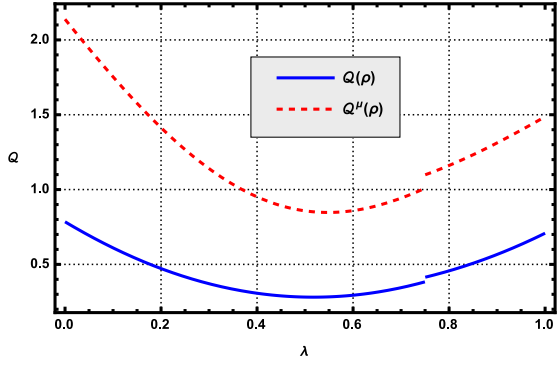}
\caption{(Color online) Total quantum correlation of the W-GHZ state as a function of $\lambda$. The discontinuity at $\lambda=\frac{3}{4}$ is because of the existence of two different $\rho_{\tilde{P}_{A_1}}$ with different quantum correlations.}
\label{W-GHZ}
\end{figure}

\subsection{Multipartite case}
A generalization of the above measure to the multipartite case is straightforward.
To do this,  we use the coefficients $C_{i_1i_2\cdots i_m}$ and construct $m(m-1)/2$ collections of matrices denoting by $T^{(A_1A_2)}_{[I_{1,2}]}$, $T^{(A_1A_3)}_{[I_{1,3}]}$, $\cdots$, $T^{(A_{m-1}A_m)}_{[I_{{m-1},m}]}$. Here,  $[I_{k,l}]$ (for $1\le k < l \le m$) stands for a collective of ($m-2$) indices, including all but the $A_k$ and $A_l$ subsystems, i.e. $I_{k,l}=i_1\cdots i_{k-1}i_{k+1}\cdots i_{l-1}i_{l+1}\cdots i_m$,  such that each index $i_n$ (for $n=1,\cdots,k-1,k+1,\cdots,l-1,l+1,\cdots,m$) ranges from $0$ to $d_{A_n}^2-1$. A general such defined matrix has the following matrix elements
\begin{equation}
T_{i_ki_l[I_{k,l}]}^{(A_kA_l)}=C_{i_1\cdots i_{k}\cdots i_{l}\cdots i_m}.
\end{equation}
We now define the $A_s$-correlation matrix $\mathcal{T}^{(A_s)}$ as
\begin{equation}
\mathcal{T}^{(A_s)}=\Big\{\vec{V}^{(A_s)}, \left[T_{i_{\tilde{s}}i_s[I_{\tilde{s},s}]}^{(A_{\tilde{s}}A_s)}\right]_{{\tilde{s}}<s}^\T , \left[T_{i_si_{\tilde{s}}[I_{s,\tilde{s}}]}^{(A_sA_{\tilde{s}})}\right]_{{\tilde{s}}>s}\Big\},
\label{tas}
\end{equation}
where $i_s$ stands for the index of the subsystem $A_s$, and $i_{\tilde{s}}$ denotes the index of the other $m-1$ subsystems $A_{\tilde{s}}$. Moreover, we should be careful when $\tilde{s}$ ranges over the $m-1$ subsystems.  In fact as it is stressed explicitly in the last two terms of the above equation, $\tilde{s}$ takes values $1,\cdots, s-1$ and  $s+1,\cdots, m$  for the second and third terms of the above equation, respectively; i.e.  terms  with and without transposition.  Finally, in constructing $\mathcal{T}^{(A_s)}$, the collective indices $I_{\tilde{s},s}$ and $I_{s,\tilde{s}}$ take  their values in such a way that   $i_{n}=0$ for $n<\tilde{s}$,  but it ranges over its appropriate values, i.e. $i_n=0,\cdots,d_{A_n}^2-1$,  for $n>\tilde{s}$.

 We are now in the position to present a necessary and sufficient condition for the classicality of a multipartite state $\rho$ with respect to a given subsystem $A_s$.
\begin{observation}
The multipartite state $\rho$ is classical with respect to part $A_s$ ($s=1,\cdots,m$) if and only if there exists a $(d_{A_s}-1)$-dimensional projection operator $P_{A_s}$ acting  on $\mathbb{R}^{d_{A_s}^2-1}$ such that
\begin{equation}
P_{A_s}\mathcal{T}^{(A_s)}=\mathcal{T}^{(A_s)},
\end{equation}
where $\mathcal{T}^{(A_s)}$ is defined by Eq. \eqref{tas}.
\end{observation}
This condition allows us to define the following measure for quantum correlation of part $A_s$, which is a multipartite generalization of the measure introduced in Ref. \cite{Akhtarshenas&etal:2015}.
\begin{equation}
\mathcal{Q}_{A_s}(\rho)=\min_{P_{A_s}}\parallel\mathcal{T}^{(A_s)}-P_{A_s}\mathcal{T}^{(A_s)}\parallel^2_2.
\label{qas}
\end{equation}
Now let $\rho$ be a multipartite state with  associated $A_s$-correlation matrices $\mathcal{T}^{(A_s)}$ defined by Eq. \eqref{tas}. Let also $P_{A_s}$ (for $s=1,\cdots,m$) represent a $({d_{A_s}-1})$-dimensional projection operator acting on the $({d_{A_s}^2-1})$-dimensional  parameter space of the subsystem $A_s$, i.e. on $\mathbb{R}^{d_{A_s}^2-1}$.
Now  in order to capture all quantum correlations of the multipartite state $\rho$, we proceed as follows. Fist,  let $\tilde{P}_{A_1}$ be the projection operator that minimizes the above distance for $s=1$, i.e.
\begin{equation}
\mathcal{Q}_{A_1}(\rho)=\parallel\mathcal{T}^{(A_1)}-\tilde{P}_{A_1}\mathcal{T}^{(A_1)}\parallel^2_2.
\end{equation}
After performing this projection operator, the $A_s$-correlation matrices transform to $\mathcal{T}_{\tilde{P}_{A_1}}^{(A_s)}$ (for $s=1,\cdots,m$), corresponding to the state $\rho_{\tilde{P}_{A_1}}$. Next, we use $\mathcal{T}_{\tilde{P}_{A_1}}^{(A_2)}$, and calculate $\mathcal{Q}_{A_2}(\rho_{\tilde{P}_{A_1}})$ as
\begin{equation}
\mathcal{Q}_{A_2}(\rho_{\tilde{P}_{A_1}})=\parallel\mathcal{T}_{\tilde{P}_{A_1}}^{(A_2)}-\tilde{P}_{A_2}\mathcal{T}_{\tilde{P}_{A_1}}^{(A_2)}\parallel^2_2,
\end{equation}
where we have assumed that $\tilde{P}_{A_2}$ is the optimized projection operator. After performing $\tilde{P}_{A_2}$, the $A_s$-correlation matrices $\mathcal{T}_{\tilde{P}_{A_1}}^{(A_s)}$ transform to $\mathcal{T}_{\tilde{P}_{A_2}\tilde{P}_{A_1}}^{(A_s)}$ (for $s=1,\cdots,m$), corresponding to the state $\rho_{\tilde{P}_{A_2}\tilde{P}_{A_1}}$. This gives us
\begin{equation}
\mathcal{Q}_{A_3}(\rho_{\tilde{P}_{A_2}\tilde{P}_{A_1}})=\parallel\mathcal{T}_{\tilde{P}_{A_2}\tilde{P}_{A_1}}^{(A_3)}-
\tilde{P}_{A_3}\mathcal{T}_{\tilde{P}_{A_2}\tilde{P}_{A_1}}^{(A_3)}\parallel^2_2.
\end{equation}
Continuing this procedure by applying the optimized projection operators $\tilde{P}_{A_3}$, $\tilde{P}_{A_4}$, $\cdots$, $\tilde{P}_{A_m}$ successively, we will obtain $A_s$-correlation matrices as $\mathcal{T}_{\tilde{P}_{A_3}\tilde{P}_{A_2}\tilde{P}_{A_1}}^{(A_s)}$,
$\mathcal{T}_{\tilde{P}_{A_4}\tilde{P}_{A_3}\tilde{P}_{A_2}\tilde{P}_{A_1}}^{(A_s)}$, $\cdots$, $\mathcal{T}_{\tilde{P}_{A_m}\tilde{P}_{A_3}\cdots\tilde{P}_{A_2}\tilde{P}_{A_1}}^{(A_s)}$. These sequences of transformed correlation matrices  can be used for the calculation of the remainder quantum correlations of the state $\rho$. A general such relation is (for $s=1,\cdots,m$)
\begin{equation}
\mathcal{Q}_{A_s}(\rho_{\tilde{P}_{A_{s-1}}\cdots\tilde{P}_{A_1}})=\parallel \mathcal{T}_{\tilde{P}_{A_{s-1}}\cdots\tilde{P}_{A_1}}^{(A_s)}
-\tilde{P}_{A_s}\mathcal{T}_{\tilde{P}_{A_{s-1}}\cdots\tilde{P}_{A_1}}^{(A_s)}\parallel^2_2,
\end{equation}
where
\begin{eqnarray}
\mathcal{T}_{\tilde{P}_{A_{s-1}}\cdots\tilde{P}_{A_1}}^{(A_s)}=\left\{\vec{V}^{(A_s)} ,
\left[\left[T_{i_{\tilde{s}}i_s[I_{\tilde{s},s}]}^{(A_{\tilde{s}}A_s)}\right]_{\tilde{P}_{A_{s-1}}\cdots\tilde{P}_{A_1}}\right]_{{\tilde{s}}<s}^\T\right.\\ \nonumber
 ,
 \left.\left[\left[T_{i_si_{\tilde{s}}[I_{s,\tilde{s}}]}^{(A_sA_{\tilde{s}})}\right]_{\tilde{P}_{A_{s-1}}\cdots\tilde{P}_{A_1}}\right]_{{\tilde{s}}>s}\right\}.
\end{eqnarray}
Here, for $k=\tilde{s}$ we have
\begin{equation}
\left[T_{i_si_{\tilde{s}}[I_{s,\tilde{s}}]}^{(A_sA_{\tilde{s}})}\right]_{\tilde{P}_{A_{k}}\cdots\tilde{P}_{A_1}}=
\sum_{i_k^{\prime}=1}^{d_{A_k}^2-1}
\left(\tilde{P}_{A_k}\right)_{i_ki_k^{\prime}}\left[T_{i_si_{\tilde{s}}^{\prime}[I_{s,\tilde{s}}]}^{(A_sA_{\tilde{s}})}\right]_{\tilde{P}_{A_{k-1}}\cdots\tilde{P}_{A_1}},
\end{equation}
while for $k\neq\tilde{s}$
\begin{equation}
\left[T_{i_si_{\tilde{s}}[I_{s,\tilde{s}}]}^{(A_sA_{\tilde{s}})}\right]_{\tilde{P}_{A_{k}}\cdots\tilde{P}_{A_1}}=
\sum_{i_k^{\prime}=1}^{d_{A_k}^2-1}
\left(\tilde{P}_{A_k}\right)_{i_ki_k^{\prime}}\left[T_{i_si_{\tilde{s}}[I_{s,\tilde{s}}^{\prime}]}^{(A_sA_{\tilde{s}})}\right]_{\tilde{P}_{A_{k-1}}\cdots\tilde{P}_{A_1}}.
\end{equation}
Note that the string $I_{s,\tilde{s}}^{\prime}$ is the same as the string $I_{s,\tilde{s}}$, except for the index $i_k$ which is changed to $i_k^{\prime}$. In addition, as it is  already mentioned for $i_k=0$, $\left[T_{i_si_{\tilde{s}}[I_{s,\tilde{s}}]}^{(A_sA_{\tilde{s}})}\right]_{\tilde{P}_{A_{k}}\cdots\tilde{P}_{A_1}}
=\left[T_{i_si_{\tilde{s}}[I_{s,\tilde{s}}]}^{(A_sA_{\tilde{s}})}\right]_{\tilde{P}_{A_{k-1}}\cdots\tilde{P}_{A_1}}$.
Then,We will find the total quantum correlation with respect to the sequence $A_1A_2\cdots A_m$ as
\begin{equation}
\mathcal{Q}_{A_1A_2\cdots A_m}=\sum_{s=1}^{m}\mathcal{Q}_{A_s}(\rho_{\tilde{P}_{A_{s-1}}\cdots\tilde{P}_{A_1}}).
\end{equation}
This can  be expressed explicitly as
\begin{equation}
\mathcal{Q}_{A_1A_2\cdots A_m}(\rho)=
\sum_{s=1}^{m}\sum_{k=d_{A_s}}^{d_{A_s}^2-1}(\tau^{(A_s)\downarrow}_{\tilde{P}_{A_{s-1}}\cdots\tilde{P}_{A_1}})_{k},
\end{equation}
where $\left\{(\tau^{(A_s)\downarrow}_{\tilde{P}_{A_{s-1}}\cdots\tilde{P}_{A_1}})_{k}\right\}_{k=1}^{d_{A_s^2-1}}$ are  eigenvalues of the matrix  $\left(\mathcal{T}_{\tilde{P}_{A_{s-1}}\cdots\tilde{P}_{A_1}}^{(A_s)}\right)\left(\mathcal{T}_{\tilde{P}_{A_{s-1}}\cdots\tilde{P}_{A_1}}^{(A_s)}\right)^{\T}$ in non-increasing order. Note that the optimization process included in the definition \eqref{qas} is now turned to
summing  the $d_{A_s}(d_{A_s}-1)$ smaller eigenvalues (counting multiplicities) of the $A_s$-correlation matrices which,  in practice, is an easy task to handle. However, in order to find $\mathcal{Q}_{A_{i_1}\cdots A_{i_m}}$ we should proceed $m$ steps as follows: In the first step, the $A_{i_1}$-correlation matrix should be constructed. Summing  up its last $d_{A_{i_1}}(d_{A_{i_1}}-1)$ non-increasing ordered eigenvalues, we will have $\mathcal{Q}_{A_{i_1}}$, while the optimized projector is $\tilde{P}_{A_{i_1}}=\sum_{j=1}^{d_{A_{i_1}}-1}n_jn_j^{\T}$ with $n_j$ being the corresponding eigenvectors of the first $d_{A_{i_1}}-1$ non-increasing ordered eigenvalues of $A_{i_1}$-correlation matrix. This procedure will be repeated for the other $(m-1)$ subsystems $A_{i_2},\cdots,A_{i_m}$, and the total quantum correlation will be calculated explicitly.  Finally, we define the total quantum correlation of the $m$-partite state $\rho$ as
\begin{equation}\label{Qfinal}
\mathcal{Q}_{\{A_1A_2\cdots A_m\}}=\max_{\mathcal{P}}\{\mathcal{Q}_{A_{i_1}A_{i_2}\cdots A_{i_m}}\},
\end{equation}
where maximum is taken over  all the $m!$ permutations of the $m$ subsystems $\{A_1A_2\cdots A_m\}$. We can define $\mathcal{Q}^\mu_{\{A_1A_2\cdots A_m\}}$ in a similar way
\begin{equation}
\mathcal{Q}^\mu_{\{A_1A_2\cdots A_m\}}=\max_{\mathcal{P}}\{\mathcal{Q}^\mu_{A_{i_1}A_{i_2}\cdots A_{i_m}}\},
\end{equation}
where $\mathcal{Q}^\mu_{A_s}(\rho_{\tilde{P}_{A_{s-1}}\cdots\tilde{P}_{A_1}})$ is defined as
\begin{equation}
\mathcal{Q}^\mu_{A_s}(\rho_{\tilde{P}_{A_{s-1}}\cdots\tilde{P}_{A_1}})=
\frac{1}{\mu[\Tr_{A_s}(\rho_{\tilde{P}_{A_{s-1}}\cdots\tilde{P}_{A_1}})]}\mathcal{Q}_{A_s}(\rho_{\tilde{P}_{A_{s-1}}\cdots\tilde{P}_{A_1}}).
\end{equation}

\section{Total quantum correlation in the one-dimensional Ising model with a transverse field}
\label{s4}

\begin{figure}
\mbox{\includegraphics[scale=0.42]{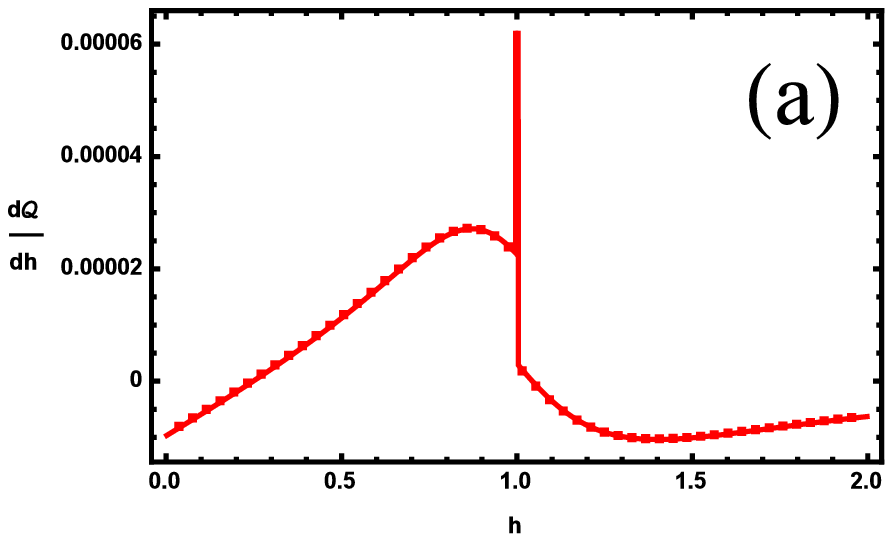}}
\mbox{\includegraphics[scale=0.42]{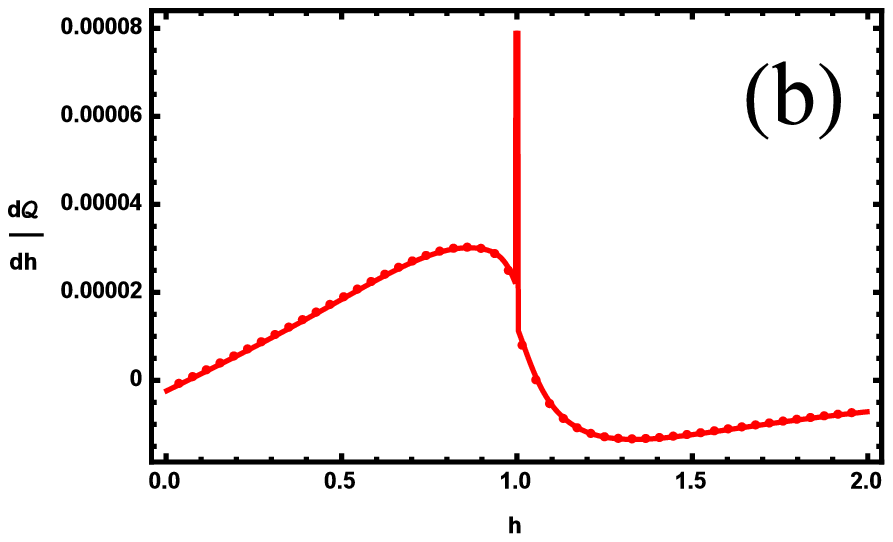}}
\mbox{\includegraphics[scale=0.42]{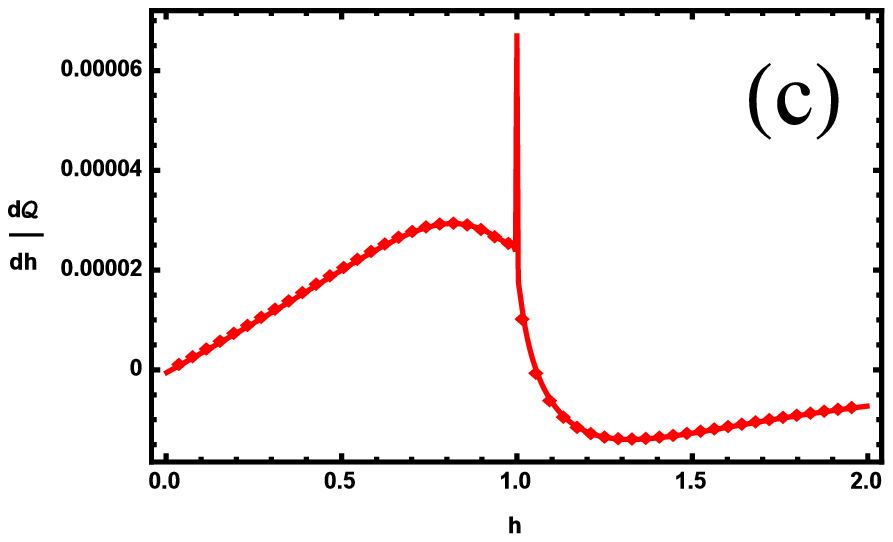}}
\mbox{\includegraphics[scale=0.42]{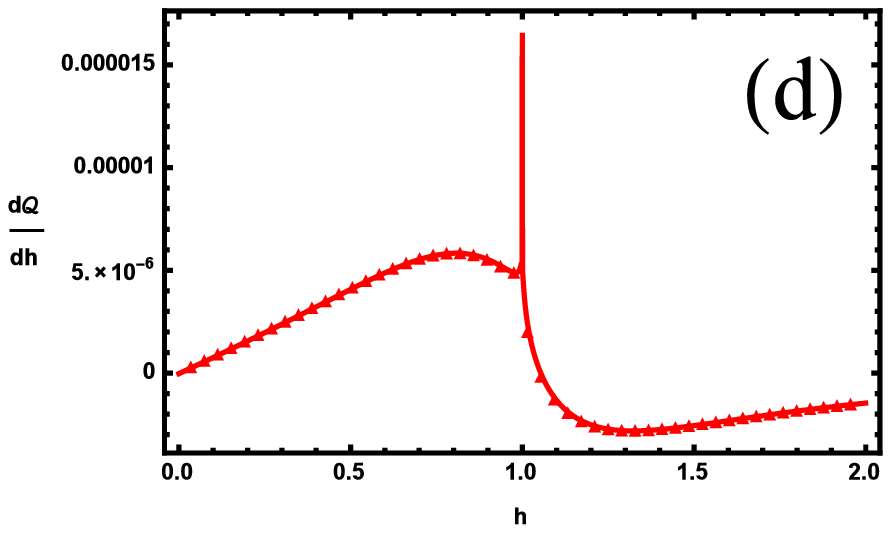}}
\caption{(Color online) First derivative of the total quantum correlation for the nearest-neighbor spins in the chains of (a) 16, (b) 64, (c) 256, and (d) 1024 spins with respect to h in the transverse field Ising chain with $\gamma=\lambda=1$.}
\label{QPT}
\end{figure}

As is mentioned in the introduction, it is reasonable to expect that a measure of quantum correlation should be able to reflect the quantum nature of states. Recently, the quantum phase transition (QPT) in the spin chains and recognizing the critical points by means of the behavior of several quantum correlation measures have been studied. Quantum discord, as a measure of quantum correlations, has been able to detect these critical points in some cases that the other kind of quantum correlations had not possessed this capability. Here, we want to  consider a model which is surveyed from different points of view, such as skew information \cite{Karpat&etal:2014} or different measures of entanglement \cite{Osborne&etal:2002,Latorre&etal:2004,Its&etal:2005,Sadiek&etal:2010,Wei&etal:2011}. The symmetric global quantum discord for this model is studied  by Sarandy {\etal} \cite{Sarandy&etal:2013}.

Consider a spin chain, consists of $m$ spin $1/2$ particles, with the coupling strength $J$ in a transverse magnetic field $h$, described by the  following Hamiltonian
\begin{equation}
H=-\frac{1}{2}J\sum_{i=1}^m\left[(1+\gamma)\sigma_i^x\sigma_{i+1}^x+(1-\gamma)\sigma_i^y\sigma_{i+1}^y\right]-h\sum_{i=1}^m\sigma_i^z,
\end{equation}
where $\sigma_i^{\alpha}$, $\alpha=\{x,y,z\}$ are the Pauli matrices and $\gamma$ is the anisotropy parameter. For the system in the thermal equilibrium, the bipartite reduced density matrix of the system is as  \cite{Sadiek&etal:2010,Zhang&etal:2012,Sarandy&etal:2013}

\begin{equation}
\rho=\left(\begin{array}{cccc}
\rho_{11}&0&0&\rho_{14}\\0&\rho_{22}&\rho_{23}&0\\0&\rho_{23}^{*}&\rho_{33}&0\\\rho_{14}^{*}&0&0&\rho_{44}
\end{array}\right),
\end{equation}
where
\begin{eqnarray}
\rho_{11}&=&\frac{1}{4}+\langle \sigma^z\rangle+\langle \sigma_k^z\sigma_l^z\rangle,\nonumber\\
\rho_{22}&=&\rho_{33}=\frac{1}{4}-\langle \sigma_k^z\sigma_l^z\rangle,\nonumber\\
\rho_{44}&=&\frac{1}{4}-\langle \sigma^z\rangle+\langle \sigma_k^z\sigma_l^z\rangle,\\
\rho_{14}&=&\langle \sigma_k^x\sigma_l^x\rangle-\langle \sigma_k^y\sigma_l^y\rangle,\nonumber\\
\rho_{23}&=&\langle \sigma_k^x\sigma_l^x\rangle+\langle \sigma_k^y\sigma_l^y\rangle.\nonumber
\end{eqnarray}
Here, the magnetization density $\langle \sigma^z\rangle$ and two-point correlation functions $\langle \sigma_k^{\alpha}\sigma_l^{\alpha}\rangle$ can be directly obtained from the exact solution of the model \cite{Pfeuty:1970,Barouch&etal:1970,Barouch&McCoy:1971}. Defining $\lambda=\frac{J}{h}$, it is shown that the model has a critical point at $\lambda=1$ which can be interpreted as a ferromagnetic-paramagnetic QPT for this chain \cite{Osborne&etal:2002,Sadiek&etal:2010}.

Now, fixing $\gamma=1$ and $J=1$, we calculate the total quantum correlation of two neighboring spins as a function of $h$ and depict the quantum correlation derivative with respect to the transverse field coupling  $h$ (see Fig. \ref{QPT}). It is clear from this figure that the critical point occurs in $h=1$,  where for the special case considering here, it is equivalent to the aforementioned situation $\lambda=1$ for the quantum critical point. It follows that our measure is able to identify the ferromagnetic-paramagnetic QPT for the model.

Although this QPT is also identified by other measures, such as global quantum discord, our  measure has the  property of being computable for any state, while the other measures of quantum correlation can be computed only for some special states. Accordingly, it seems that  our measure can be a useful tool to detect the quantum critical points for the states whose other measures  can not be calculated analytically. It is worth to mention that, for this special example, $\mathcal{Q}^{\mu}$ and $\mathcal{Q}$ show similar manners and predict the QPT in the same point, so  we just illustrated  $\mathcal{Q}$ for simplicity.

\section{Conclusion}
\label{s5}
In this paper we have generalized the notion of $A$-correlation matrix of a bipartite system $\mathcal{H}^{(A)}\otimes\mathcal{H}^{(B)}$ to the multipartite case $\mathcal{H}^{(A_1)}\otimes\mathcal{H}^{(A_2)}\otimes\cdots\otimes\mathcal{H}^{(A_m)}$, and presented a necessary and sufficient condition for classicality of such system with respect to the subsystem $A_s$.  Based on this, we have presented a computable measure of the total quantum correlation of a generic $m$-partite state. Some illustrative examples, both in the bipartite and tripartite systems, have been also  presented and a comparison with the other measures has been investigated. We have also considered a spin chain and investigated the ability of the measure to detect quantum critical points.   This paper, therefore, can be regarded as a further development in the calculation of quantum correlations of multipartite systems which can be used as an indicator for the quantumness of  multipartite systems. Further study on the measure, and  the role of the $p$-norm, as well as the arbitrariness in choosing the function of the coherence vector of the unmeasured subsystems, are under considerations.

\bibliography{BIB}{}

\end{document}